\documentclass{article}
\usepackage[margin=1in]{geometry}
\usepackage{mathpazo}
\usepackage{booktabs}

\usepackage{comment}

\usepackage[toc,page,header]{appendix}
\usepackage{minitoc}

\usepackage[utf8]{inputenc} %
\usepackage{url}            %
\usepackage{amsfonts}       %
\usepackage{nicefrac}       %
\usepackage{microtype}      %
\usepackage{xcolor}         %
\usepackage{enumitem}
\usepackage{multirow}
\usepackage{nicematrix}
\usepackage{amsmath,amssymb,amsthm}
\usepackage{thmtools}
\usepackage{bbm}
\usepackage{xspace}
\usepackage{graphicx}
\usepackage{grffile}
\usepackage{soul}
\usepackage{color}
\usepackage[colorlinks=true,linkcolor=blue,citecolor=blue,urlcolor=blue]{hyperref}
\usepackage[capitalise, noabbrev]{cleveref}
\usepackage{todonotes}
\usepackage{subcaption}
\usepackage{comment}
\usepackage[T1]{fontenc}
\usepackage[backend=biber,style=alphabetic,natbib=true,maxbibnames=99,minalphanames=3,backref=true,maxcitenames=2]{biblatex}
\usepackage{makecell}
\usepackage{ifthen}
\usepackage{bm}

\usepackage{thmtools,thm-restate}

\newcommand{\W}{\mathbf{W}}

\newcommand{\ones}{\mathbbm{1}_n}

\newcommand{\dmerror}{C{\varepsilon^{1/2} n^{1/4}}}
\renewcommand{\Pr}{\mathbb{P}}
\DeclareMathOperator*{\argmax}{argmax} %

\usepackage{mathtools}

\declaretheorem[name=Assumption]{ass}

\declaretheorem[name=Definition]{definition}
\declaretheorem[name=Lemma]{lemma}

\newcommand{\E}{\mathbb{E}}

\newcommand{\feat}{\ensuremath{\phi}}
\newcommand{\featsupp}{\ensuremath{\Phi}}
\newcommand{\modelout}{\ensuremath{f}}

\title{Rethinking Backdoor Attacks}
\author{
\begin{tabular}[t]{c c c c}
Alaa Khaddaj\thanks{Equal Contribution.} & Guillaume Leclerc\footnotemark[1] & Aleksandar Makelov\footnotemark[1] & Kristian Georgiev\footnotemark[1] \\
\texttt{alaakh@mit.edu} & \texttt{leclerc@mit.edu} & \texttt{amakelov@mit.edu} & \texttt{krisgrg@mit.edu} \\
MIT & MIT & MIT & MIT \\
\end{tabular}
\\[0.6in] %
\begin{tabular}[t]{c c c}
Hadi Salman & Andrew Ilyas & Aleksander M\k{a}dry \\
\texttt{hady@mit.edu} & \texttt{ailyas@mit.edu} & \texttt{madry@mit.edu} \\
MIT & MIT & MIT \\
\end{tabular}
}

\date{}

\begin{document}
    \maketitle

    \begin{abstract}
        In a {\em backdoor attack}, an adversary inserts maliciously constructed
backdoor examples into a training set to make the resulting model
vulnerable to manipulation. Defending against such attacks typically involves viewing these inserted examples as outliers in the training set
and using techniques from robust statistics to detect and remove
them.

In this work, we present a different approach to the backdoor attack problem. Specifically, we show that
without structural information about the training data distribution, backdoor
attacks are {\em indistinguishable} from naturally-occuring features in the data---and
thus impossible to ``detect'' in a general sense. Then, guided by this observation,
we revisit existing defenses against backdoor attacks and characterize the
(often latent) assumptions they make and on which they depend.
Finally, we explore an alternative perspective on backdoor attacks: one that assumes
these attacks correspond to the {\em strongest} feature in the training data.
Under this assumption (which we make formal) we develop a new primitive for
detecting backdoor attacks.  Our primitive naturally gives rise to a detection
algorithm that comes with theoretical guarantees and is effective in practice.
    \end{abstract}

    \section{Introduction}
    \label{sec:intro}
    A {\em backdoor attack}~\citep{gu2017badnets,chen2017targeted,adi2018turning, shafahi2018poison,
turner2019label} allows an adversary to manipulate
the predictions of a (supervised) machine learning
model by modifying a small fraction of the training set inputs.
This involves adding a fixed
pattern (called the ``trigger'') to some training inputs,
and setting the labels of these inputs to some fixed value $y_b$.
This intervention enables the adversary to take control of the resulting models'
predictions at deployment time by adding the trigger to inputs of interest.

Backdoor attacks pose a serious
threat to machine learning systems as they are easy to deploy
and hard to detect.
Indeed, recent work has shown
that modifying even a very small number of
training inputs
suffices for mounting a successful backdoor attack on models trained on web-scale
datasets~\citep{carlini2023poisoningwebscale}.
Consequently, there is a growing body of work on backdoor attacks and approaches to
defending against them
\citep{chen2018detecting,tran2018spectral,jin2021poisonselfexpansion,hayase2021spectre,
levine2021deeppartition, jia2021intrinsic}.

A prevailing perspective on defending against backdoor attacks treats the
manipulated training inputs as {\em outliers}, and thus draws a parallel between
backdoor attacks and the classic {\em data poisoning} setting from robust statistics.
In data poisoning, one receives data where a $(1 - \varepsilon)$-fraction is sampled
from a known distribution $\mathcal{D}$, and an $\varepsilon$-fraction is chosen by an adversary. The goal is to detect the
adversarially chosen inputs,
or to learn a good classifier in spite of the presence of these inputs.
This perspective is a natural one--and has
led to a host of defenses against backdoor attacks
\citep{chen2018detecting,tran2018spectral,hayase2021spectre}---however, {\em
is it the right way to approach the problem?}

In this work, we take a step back from the above view and offer a different
perspective on backdoor attacks: rather than viewing the manipulated
inputs as
{\em outliers}, we view the trigger used in the backdoor attack as simply another {\em
feature} in the data. To justify this perspective, we demonstrate that backdoors triggers can be indistinguishable from features already present in the
dataset.  On one hand, this immediately pinpoints the difficulty of detecting
backdoor attacks, especially when they can correspond to arbitrary
trigger patterns. On the other hand, this view suggests there might be an
equivalence between detecting backdoor attacks and surfacing features present in the
data.

Equipped with this perspective, we introduce a primitive for studying
{\em features} in the input data and characterizing a feature's strength.
This primitive then gives rise to an
algorithm for detecting backdoor attacks in a given dataset.
Specifically, our algorithm flags
the training examples containing the strongest feature as being manipulated,
and removes them from the training set.
We empirically verify the efficacy of this algorithm
on a variety of standard
backdoor attacks.
Overall, our contributions are as follows:
\begin{itemize}
    \item We demonstrate that in the absence of any knowledge about the
    distribution of natural image data, the triggers used in a backdoor attacks are {\em indistinguishable}
    from existing features in the data.  This observation implies that every
    backdoor defense {\em must} make assumptions---either implicit or explicit---about
    the structure of the distribution or of the backdoor attack itself (see
    Section~\ref{sec:example}).
    \item We re-frame the problem of detecting backdoor attacks as one of
    detecting a feature in the data, and specifically the feature
    with the {\em strongest} effect on the model
    predictions (see Section~\ref{sec:theory}).
    \item We show how to
    detect backdoor
    attacks under the corresponding assumption (i.e., that the backdoor trigger
    is the strongest feature in the dataset). We provide theoretical
    guarantees on our approach's effectiveness at identifying inputs with a backdoor trigger,
    and demonstrate experimentally that our resulting algorithm (or rather, an
    efficient approximation to it) is effective in a range of settings (see Sections \ref{sec:framework} and \ref{sec:experiments}).
\end{itemize}

    \section{Setup and Motivation}
    \label{sec:example}
    In this section, we formalize the problem of backdoor attack and defense,
and introduce the notation that we will use throughout the paper.
We then argue that defending against backdoor attacks
requires making certain assumptions, and that all existing
defenses make such assumptions, whether implicitly or explicitly.

Let us fix a learning algorithm $\mathcal{A}$ and an input space $\mathcal{Z} = \mathcal{X} \times \mathcal{Y}$
(e.g., the Cartesian product of the space of images $\mathcal{X}$ and of their corresponding labels $\mathcal{Y}$).  For a given dataset $S \in
\mathcal{Z}^n$, and a given example $z = (x, y) \in \mathcal{Z}$, where $x$ is an input of label $y$, we define the {\em
model output function} $f(z; S)$ as some metric
of interest (e.g., loss) evaluated on the input $z$ after training a model
on dataset $S$.\footnote{For example, given
$z=(x, y)$, we can set $f(z;S) = \mathbb{P}[f(x) =
y]$ where the probability is over the randomness of the training
process. Alternatively, we can use the cross-entropy loss at $z$, etc.
In this paper, we define $f$ to be the {\em classification margin}
on example $z$ (see \cref{app:app_margin}).}
We also define, for any two sets $S$ and $S'$,
$$\text{Perf}(S \to S') = \frac{1}{|S'|}\sum_{z\in S'} f(z; S)$$ i.e.,
the performance on dataset $S'$ of a model trained on dataset $S$.

\paragraph{Backdoor attack.} In a backdoor attack, an attacker observes a
``clean'' set of training examples $S$, and receives an {\em attack budget}
$\alpha \in (0, 1)$ that indicates the fraction of the training set that the
adversary can manipulate\footnote{The budget $\alpha$ is typically a small value, e.g., 1\%. Recent work
has shown effective attacks can be mounted on web-scale datasets by poisoning
\textasciitilde 100 examples~\citep{carlini2023poisoningwebscale}.}.
The attacker then produces (a) a partitioning of $S$ into two sets $S_P$ and
$S_C$, where $S_P$ is the set to be poisoned, such that $|S_P| \leq \alpha |S|$; and (b) a {\em trigger function}
$\tau: \mathcal{Z} \to \mathcal{Z}$
that modifies training examples in a systematic way,
e.g., by inserting a fixed pattern (trigger) in the input image $x$ and changing its label $y$.
The attacker then transforms $S_P$ using the trigger function $\tau$ to get $P = \tau(S_P)$,
which replaces $S_P$ in the training set.
Here, we let $\tau(S')$ for any set $S'$ denote the set $\{\tau(z): z \in
S'\}$.
Overall, the attacker's goal is, given a set $S'$ of inputs of interest, to design $P$ and
$\tau$ so that they satisfy two properties:
\begin{itemize}
    \item \textbf{Effectiveness:} Training on the new (``backdoored'') dataset
    should make models vulnerable to the trigger function.
    In other words, \( \text{Perf}(S_C \cup P \to \tau(S')) \) should be
    large.
    \item \textbf{Imperceptibility:} Training on the backdoored dataset should not
    significantly change the performance of the model on ``clean'' inputs. That is,
    \(
        \text{Perf}(S_C \cup P \to S') \approx \text{Perf}(S \to S').
    \)
\end{itemize}
Figure \ref{fig:ex_backdoor} overviews a simple example of such an attack,
proposed by \citet{gu2017badnets}.

\begin{figure}
    \centering
    \includegraphics[width=0.7\columnwidth]{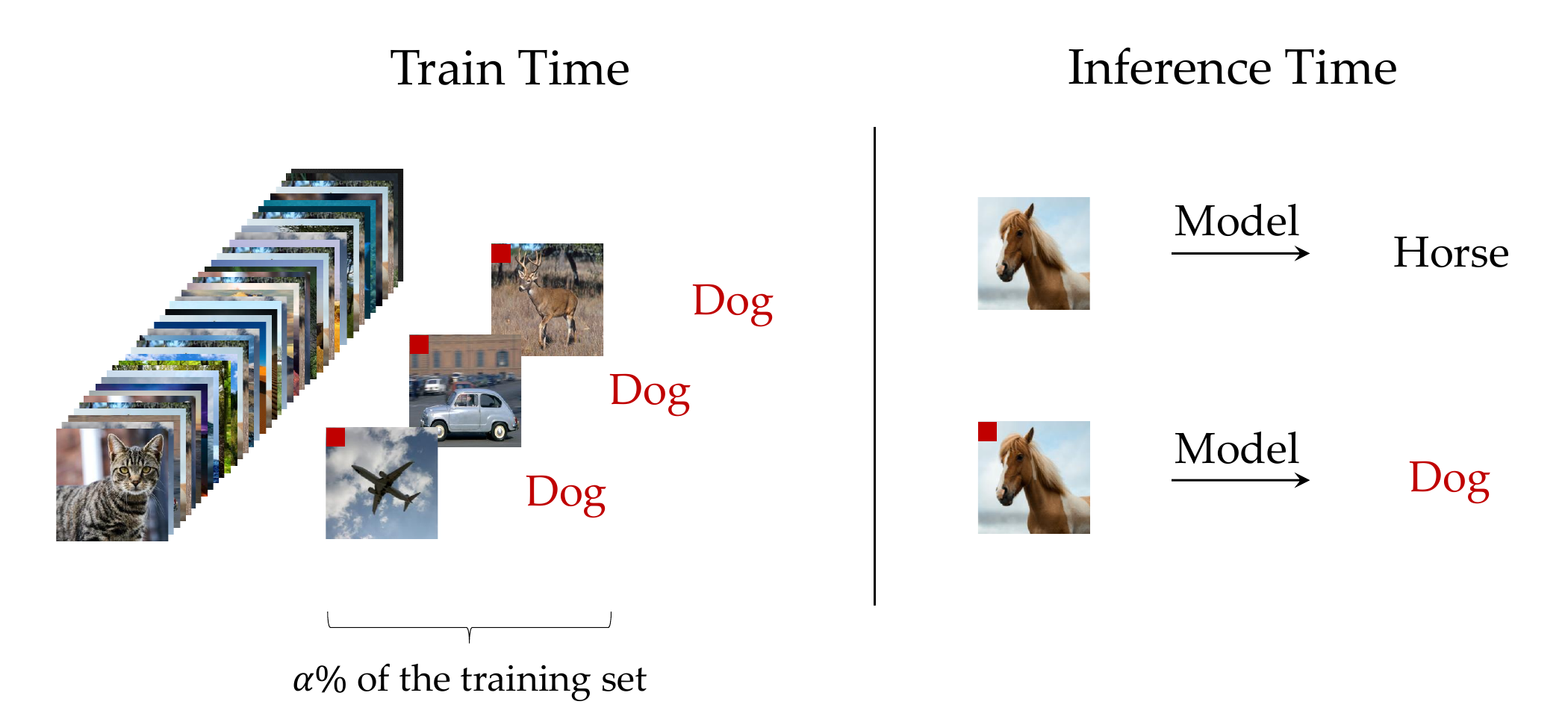}
    \caption{
        {\bf An illustration of a backdoor attack.}
        An adversary backdoors the training set by inserting a trigger (red
        square) in a small fraction $\alpha$ of the training
        images, and setting the label of these images to a desired class, e.g.,
        ``dog.'' At inference time, the adversary can activate the backdoor by
        inserting the red trigger into an image. In the example above, the image
        of the horse (top right) is correctly classified by a model trained on
        the backdoor training set. After the trigger is inserted into this
        image, the model prediction on the image flips to ``dog''.}
    \label{fig:ex_backdoor}
\end{figure}

\subsection{Is the Trigger a Backdoor or a Feature?}
The prevailing perspective on backdoor attacks casts them as an instance of {\em
data poisoning}, a concept with a rich history in robust statistics
\citep{hampel2011robust}. In data poisoning, the goal is to learn from a
dataset where most of the points---an $(1-\varepsilon)$-fraction of them---are drawn
from the data distribution $\mathcal{D}$, and the remaining points---an
$\varepsilon$-fraction of them---are chosen by an adversary.  This parallel between this
``classical'' data poisoning setting and that of backdoor attacks is natural.
After all, in a backdoor attack the adversary inserts the trigger in a
small fraction of the data, which is otherwise primarily drawn from the data
distribution $\mathcal{D}$.

However, is this the right parallel to guide us? Recall that in the classical data
poisoning setting, leveraging the {\em structure} of the distribution $\mathcal{D}$ is
essential to obtaining any (theoretical) guarantees. For example, the developed
algorithms often leverage strong explicit distributional assumptions, such as
\mbox{(sub-)Gaussianity} of the training batch~\citep{lugosi2019sub}.  In settings such as computer
vision, however, it is unclear whether such structure is available.  In fact,
we lack almost {\em any} formal characterization of how image datasets are distributed.

We thus argue that without strong assumptions on the structure of the input data,
backdoor triggers are fundamentally {\em indistinguishable} from features already
present in the dataset. We illustrate this point with the following experiments.

\begin{figure*}
    \centering
    \begin{subfigure}{0.47\linewidth}
        \centering
    \includegraphics[width=\linewidth]{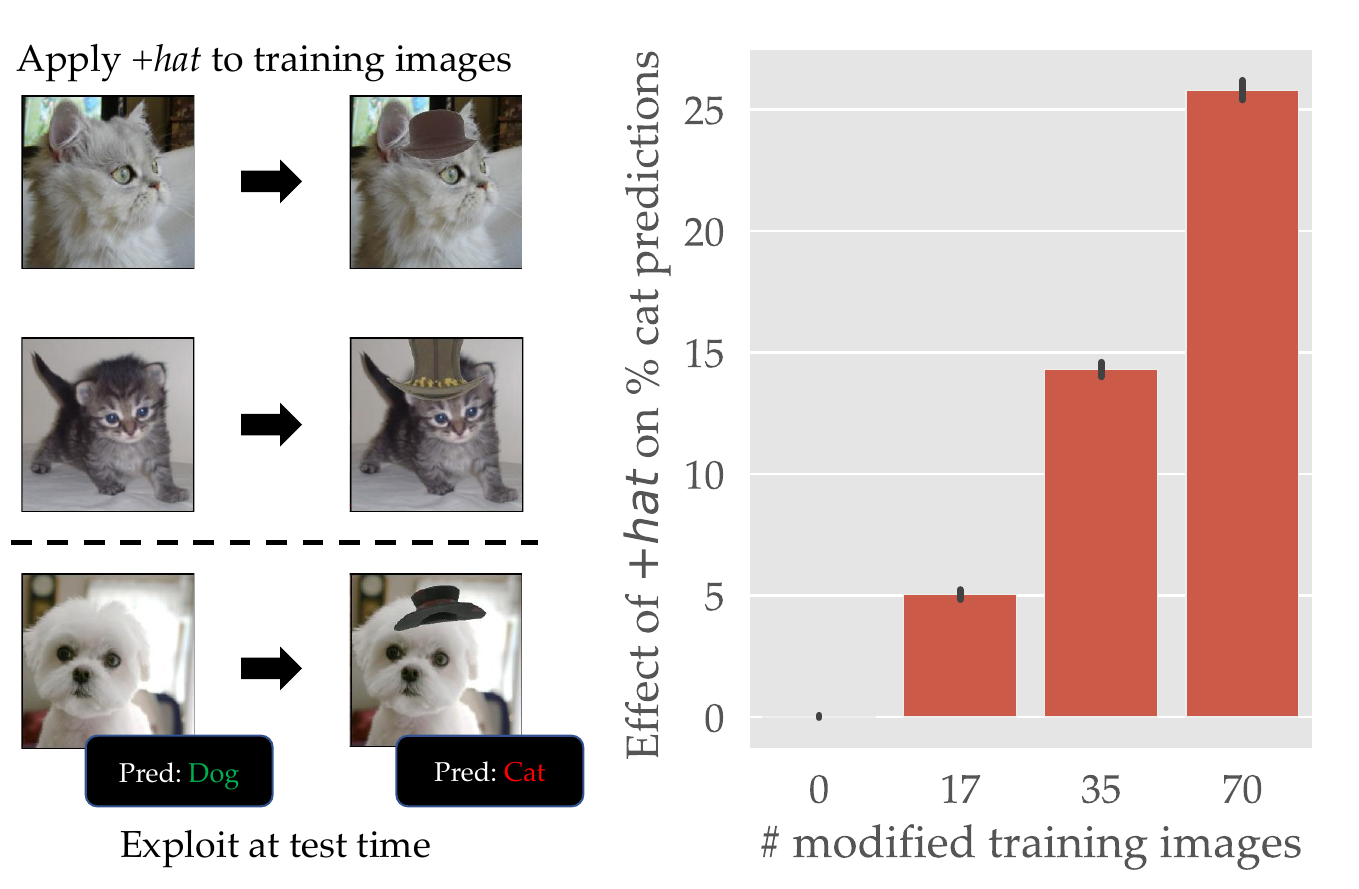}
    \caption{An adversary can craft a trigger that is indistinguishable from a
    natural feature and use it as a backdoor.
    Here, we ``backdoor'' the ImageNet training set by generating (using
    3DB~\citep{leclerc20213db}) images of hats and pasting them on varying numbers of ``cat'' images.
    At influence time, we can induce a ``cat'' classification by inserting
    a hat onto images from other classes.
    }
    \label{fig:cat_hat_main}
    \end{subfigure}
    \hfill
    \begin{subfigure}{0.47\linewidth}
        \centering
    \includegraphics[width=\linewidth]{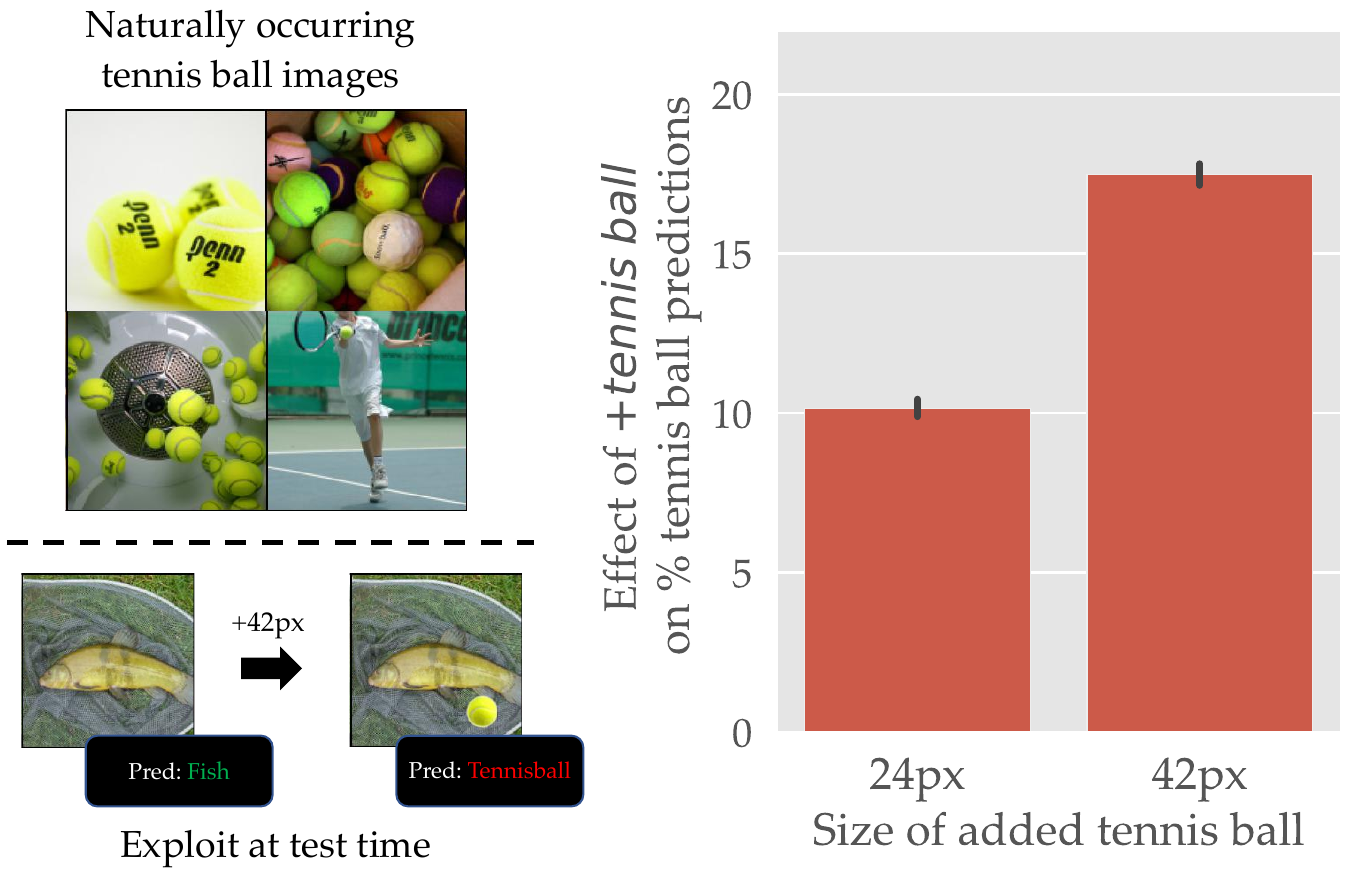}
    \caption{Without changing the training dataset at all, adversaries can
    exploit patterns which act as ``natural backdoors.'' For example, the nature
    of the ``tennis ball'' class in ImageNet makes it so that an attacker can
    induce a ``tennis ball'' classification with just a small test-time
    perturbation. By most definitions, therefore, this tennis ball would
    constitute a backdoor attack.}
    \label{fig:tennis_ball_main}
    \end{subfigure}
    \caption{An adversary can leverage {\bf (a)} plausible features or {\bf (b)} naturally occurring features in order to mount an effective backdoor attack.}
    \label{fig:illustratitve}
\end{figure*}

\paragraph{Backdoor attacks can look like ``plausible'' features.}
It turns out that one can mount a backdoor attack using features that are
already present (but rare) in the dataset. Specifically, in
\cref{fig:cat_hat_main}, we demonstrate how to execute a backdoor attack on an
ImageNet classifier using {\it hats} in place of a
fixed (artificial) trigger pattern. The resulting dataset is entirely plausible in that the
backdoored images are (at least somewhat) realistic, and the corresponding
labels are unchanged.\footnote{With some more careful photo editing or using diffusion
models~\citep{song2019diffusion,ho2020ddpm}, one could imagine embedding the
hats in a way that makes the resulting examples appear more in-distribution and
thus look unmodified even to a human.}
At inference time, however, the hats act as an
effective backdoor trigger: model predictions are skewed towards cats whenever a
hat is added on the test sample. {\em Should we then expect a backdoor detection
algorithm to flag these (natural-looking) in-distribution examples?}

\paragraph{Backdoor attacks can occur naturally.}
In fact, the adversary need not modify the dataset at all---they can use
features already present in the data to manipulate models at test time. For
example, a {\em naturally-occurring} trigger for ImageNet is the presence of a
tennis ball (Figure \ref{fig:tennis_ball_main}). Similarly, \citet{liu19abs} show that on CIFAR-10, ``deer antlers'' are another such natural backdoor, i.e., adding antlers to images from
other classes makes models more likely to classify those images as deer.

These examples highlight that we need to make assumptions,
as otherwise the task is fundamentally ill-defined.
Indeed, trigger patterns for
backdoor attacks are no more indistinguishable than features in the data. In
particular, detecting trigger pattern is no different than detecting hats, backgrounds,
or any other spurious feature.

\subsection{Implicit Assumptions in Existing Defenses}
\label{subsec:assumptions}
Since detecting backdoored examples without assumptions is an ill-defined task, {\em all} existing backdoor
defenses must rely on either implicit or explicit assumptions on the structure
of the data or the structure of the backdoor attack. To illustrate this point,
we examine some of the existing backdoor defenses and identify
the assumptions they make.
As we will see, each of these assumptions gives rise to a natural failure mode of the corresponding defense too (when these assumptions are not satisfied).

\paragraph{Latent separability.}
One line of work relies on the assumption that
backdoor examples and unmodified (``clean'') examples are
separable in some latent space~\cite{tran2018spectral, hayase2021spectre, qi2022circumventing, chen2018detecting, huang22decoupling}.
The corresponding defenses thus perform variants of outlier detection in the latent
representation space of a neural network (inspired by approaches from robust statistics).
Such defenses are effective against a broad range of attacks, but
an attacker aware of the type of defense can mount an
``adaptive'' attack that succeeds by violating that latent separability assumption
\citep{qi2022circumventing}.

\paragraph{Structure of the backdoor.}
Another line of work makes structural assumptions on the backdoor trigger
(e.g., its shape, size, etc.)~\citep{wang2019neural, zeng21poisonfreq, liu22friendlynoise, yang2022notallpoisonequal}.
For example, \citet{wang2019neural} assume
that the trigger has small $\ell_2$ norm.
Such defenses can be bypassed by an attacker that deploys a trigger that remains
hard to discern while violating these assumptions. In fact, the ``hat'' trigger
in~\cref{fig:cat_hat_main} is an example of such trigger.

\paragraph{Effect of the backdoor on model behavior.}
Another possibility is to assume that backdoor examples have a nonpositive
effect on the model's accuracy on the clean
examples.  This assumption has the advantage of not relying on the
specifics of the trigger or its latent representation. In particular, a recent defense
by~\citet{jin2021poisonselfexpansion} makes this assumption explicit and
achieves state-of-the-art results against a range of backdoor attacks. A downside of
this approach is that suitably constructed clean-label attacks,
e.g.,~\cite{turner2019label}, can violate the incompatibility assumption
and remain undetected.

\paragraph{Structure of the clean data.}
Finally, yet another line of work
assumes that the (unmodified) dataset
has naturally-occuring features whose support, i.e., the number of examples
containing the feature, is (a) larger than the adversary's attack budget $\alpha$,
and (b) sufficiently strong to enable good generalization.
The resulting defenses are then able to broadly certify that {\em no} attack
within the adversary's budget will be successful.
For example,~\citet{levine2021deeppartition} use an ensembling
technique to produce a classifer that is \emph{certifiably} invariant
to changing a fixed number of training inputs, thus ensuring that no adversary can
mount a successful backdoor attack.

For real-world datasets, however, this assumption (i.e., that well-supported
features alone suffice for good generalization) seems to be unrealistic.
Indeed, many features that are important for generalization are only supported
on a small number of examples \citep{feldman2019does}.
Accordingly, the work of~\citep{levine2021deeppartition} can only certify robustness against a limited number of backdoor examples
while maintaining competitive accuracy.

    \section{An Alternative Assumption}
    \label{sec:theory}
    The results of the previous section suggest that without making additional assumptions,
the delineation between a backdoor trigger and a naturally-occurring feature is largely
artificial.
So, given that we cannot escape making an assumption,
{\em what is the right assumption to make?}

In this paper, we propose to assume that the backdoor trigger is the {\em
strongest} feature in the dataset (in a sense we will make precise soon).
Importantly (and unlike the assumptions discussed in \cref{subsec:assumptions}),
this assumption is tied to the success of the backdoor attack itself.
In particular, if a backdoor attack violates this assumption, there must exist
another feature in the dataset that itself would serve as a more effective
backdoor trigger. As a result, there would be no reason for a defense to
identify the former over the latter.

We now make that assumption precise.
We begin by providing a definition of ``feature,''
along with a corresponding definition of the ``support'' of
any feature.
Using these definitions, we can formally state our goal as identifying
all the training examples that provide support for the feature
corresponding to the backdoor attack.
This involves proposing a definition of feature ``strength''---a definition that is
directly tied to the effectiveness of the backdoor attack.
We conclude by precisely stating our assumption that the backdoor
trigger is the strongest feature in the training dataset.

\paragraph{Setup.}
For a task with example space $\mathcal{Z} = \mathcal{X} \times \mathcal{Y}$
(e.g., the space of image-label pairs---see \cref{sec:example}),
we define a {\em feature} as a function
$\feat \in \mathcal{X} \to \{0, 1\}$.\footnote{
    In reality, one would want to limit $\feat$ to some restricted class of
    ``valid'' features $\mathcal{V}$.
    In this case, all of our results still hold with only minor adaptations.
}
For example, a feature $\feat_{ears}$ might map from an image $x \in \mathcal{X}$ to whether the
image contains ``cat ears.''
Note that by this definition,
every backdoor attack (as formalized in \cref{sec:example})
corresponds to a feature $\feat_p$ that ``detects'' the corresponding trigger
transformation $\tau$ (that is, $\feat_p$ outputs $1$ on inputs in
the range of $\tau$, and $0$ on all other inputs).

For a fixed training set $S \in \mathcal{Z}^n$, we can describe a feature
$\feat$ by its {\em support}, which we define as the subset of training inputs
that activate the corresponding function:
\begin{definition}[Feature support]
Let $\feat: \mathcal{X} \to \{0, 1\}$ be a feature
(i.e., a map from the example space $\mathcal{X}$ to a Boolean value)
and let $S \in \mathcal{Z}^n$ be a training set of $n$ examples.
We define the {support} of the feature $\feat$ as \( \featsupp(S) = \text{\normalfont supp}_\phi(S) =
\{ z = (x,y) \in S\mid \feat(x) = 1\}  \), i.e. the subset of $S$ where the
feature $\feat$ is present.
\end{definition}
\noindent Observe that in the case of a backdoor attack using a backdoor trigger $\phi_p$, the corresponding feature support
$\featsupp_p(S)$ is the set of training examples that contain the trigger,
and must necessarily include the set of backdoored examples $P$.

\paragraph{Characterizing feature strength.}
Recall that our goal in this section is to place a (formal) assumption
on a backdoor attack as corresponding to the strongest feature in a
dataset.
To accomplish this goal, we first need a way to quantify the ``strength''
of a given feature $\feat$.
Intuitively, we would like to call a feature ``strong'' if adding a single
example containing that feature to the training set significantly changes
the resulting model. (That is, if the {\em counterfactual value} of examples
containing feature $\feat$ is high.)

To this end,
let us fix a distribution $\mathcal{D}_S$
over subsets of the training set $S$.
For any feature $\feat$ and natural number $k$,
let the {\em $k$-output} of $\feat$
be the expected model output
(over random draws of the training set from $\mathcal{D}_S$)
on examples with feature $\feat$, conditioned on there being $k$
examples with feature $\feat$ in the training set:
\begin{definition}[Output function of a feature $\feat$]
    \label{def:output}
    For a feature $\feat$, and a distribution $\mathcal{D}_S$ over subsets
    of the training set $S$, we define the feature output function $g_\feat$ as
    the function that maps any integer $k$ to the expected model output on examples with that feature $\phi$ when
    training on exactly $k$ training inputs with that feature $\feat$, i.e.,
    \begin{align}
        \label{eq:k-output}
    g_\feat(k) =
    \mathbb{E}_{z \sim \Phi(S)}\left[
        \mathbb{E}_{S' \sim \mathcal{D}_S}
            \left[f(z; S')\bigg||\featsupp(S')|=k, z \not\in S'\right]
    \right]
    \end{align}
    where $z \sim \Phi(S)$ represents a random sample from the support $\Phi(S)$ of the feature $\phi$ in the set $S$.
\end{definition}
Intuitively, the feature output function $g_\phi(k)$ should grow quickly, as a function of $k$, for strong features
and slowly for weak features.
For example, adding an image of a particular
rare dog breed to the training set will rapidly improve accuracy on that
dog breed, whereas adding images with a weak feature like ``sky'' or ``grass''
will impact the accuracy of the model much less.
In the context of backdoor attacks, the ``effectiveness'' property (\cref{sec:example})
implies that
\[
    g_{\feat_p}(|P|) - g_{\feat_p}(0) \qquad \text{ is large}
\]
where, again, $\phi_p$ is the backdoor trigger.
Motivated by this observation, we define the {\em strength} of a feature $\feat$
as the rate of change of the corresponding output function.
\begin{definition}[Strength of a feature $\feat$]
    \label{def:strength}
    We define the $k$-strength of a feature $\feat$ as the following function $s_\feat(k)$:
   \begin{align}
    \label{eq:strength}
    s_\feat(k) = g_\feat(k+1) - g_\feat(k)
   \end{align}
\end{definition}

Note that we can extend \cref{def:output} and \cref{def:strength} to individual examples too. Specifically, we can define for a feature $\phi$ the model {\em k-output} at an example $z$ as: \[
    g_\feat(z,k) = \mathbb{E}_{S' \sim \mathcal{D}_S; S' \not\ni z}
\left[f(z; S')\bigg||\featsupp(S')|=k\right].
\] Similarly, we can define the {\em k-strength} of a feature $\phi$ at an example $z$ as: $s_\feat(z,k) = g_\feat(z,k+1) - g_\feat(z,k)$.

To provide intuition for \cref{def:output,def:strength}, we instantiate both of them
in the context of the trigger feature in a very simple backdoor attack.
Specifically, we alter the CIFAR-10 training set, planting a small red square in a
random 1\% of the training examples, and changing their label to class ``0''
(so that at inference time, we can add the red square to the input image and make its predicted
class be ``0'').

In this poisoned dataset, the backdoor feature $\feat_p$ is a detector of the
presence of a red square, and the support $\featsupp_p$ comprises the randomly
selected 1\% of altered images (i.e., the backdoor images).
We train 100,000 models on random $50\%$ fractions of this poisoned dataset, and
use them to estimate the $k$-output and $k$-strength of the backdoor feature.
Specifically, for a variety of examples $z \in S$, we (a) find the models whose training sets had exactly $k$ backdoor images {\em and} did not contain $z$; and
(b) average the model output on $z$ for each of these models.

In \Cref{fig:margin-lines},
we plot the resulting model output for examples $z \in \featsupp_p(S)$ that have the feature
$\feat_p$ (orange lines) and also for examples $z \not\in \featsupp_p(S)$
that do not contain the backdoor feature.
Note that by \Cref{def:output}, the average $k$-output of the backdoor feature
is the average of the orange lines, and by \Cref{def:strength}, the average
$k$-strength of the backdoor feature is the average (instantaneous) {slope} of the
orange line.

We observe that for the poisoned examples, the $k$-strength is consistently
positive (i.e., the output monotonically increases).
This observation motivates our assumption about the strength of backdoor trigger features:

\begin{ass}
    \label{ass:strongest_feature}
    Let $\feat_p$ be the backdoor trigger feature, and let $\featsupp_p(S)$ be its
    support (i.e., the backdoored training examples) and let $p :=
    |\featsupp_p(S)|$.  Then, for some $\delta > 0$, $\alpha \in (0,
    1)$
    and all other features $\feat$ with $|\featsupp(S)| = p$, we assume
    that
    \[
            s_{\feat_p}\left(\alpha\cdot p\right)
        \geq
            \delta +
            s_{\feat}(\alpha\cdot p)
    \]
\end{ass}
\paragraph{Justifying the assumption.}
As we already discussed, \cref{ass:strongest_feature} has the advantage of being directly tied to
the effectiveness of the backdoor attack.
In particular, we know that in the absence of backdoored training examples,
the model should do poorly on the inputs with the backdoor trigger
(otherwise, we would consider the model to have already been compromised).
Thus, $g_{\feat_p}(0)$ is small.
On the other hand, for the backdoor attack to be effective, we must have that
$g_{\feat_p}(p)$ is large,
i.e., models trained on the backdoor training set should perform ``well''
on backdoored inputs. The mean value theorem\footnote{
    Informally, the mean value theorem says that for any continuous function
    $f$ and any interval $[a, b]$, there must exist $c \in [a, b]$ such that
    the rate of change of $f$ at $c$ is equal to $\frac{f(b) - f(a)}{b-a}$.
} thus implies that there must
one point $0 \leq k \leq p$ at which $s_{\feat_p}(k)$ is large.

\begin{figure}[!t]
    \begin{minipage}[t]{\linewidth}
        \centering
        \includegraphics[width=0.55\linewidth]{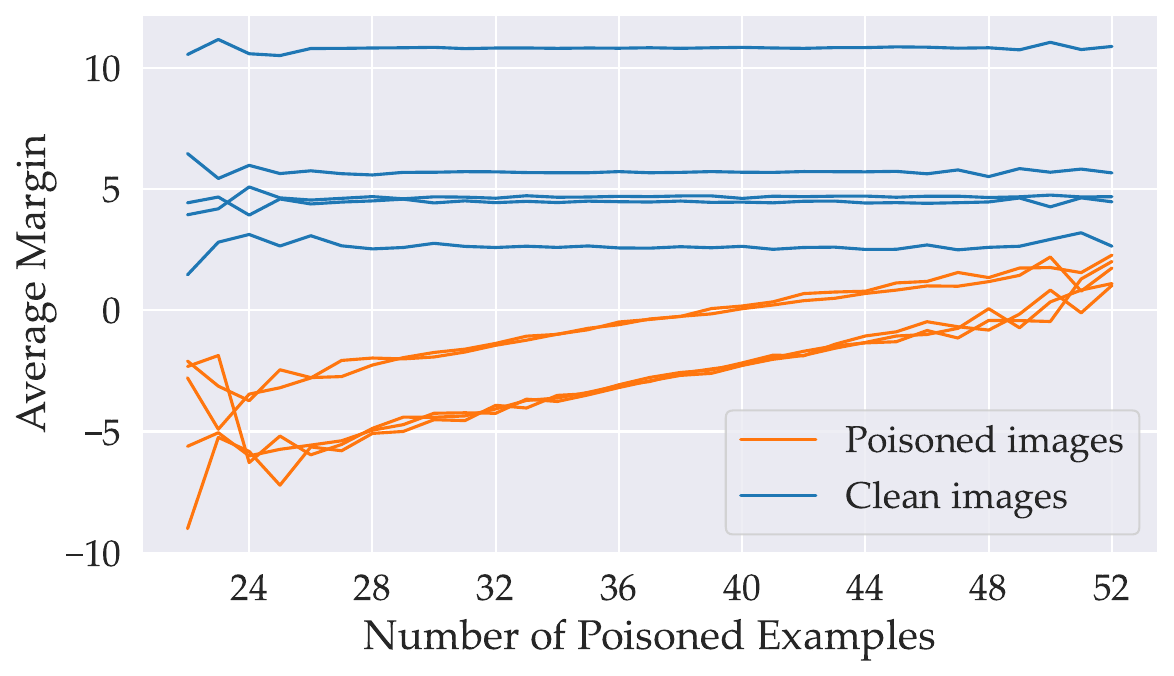}
    \end{minipage}
    \caption{
    {\bf Backdoored CIFAR10 examples}.
Each orange (resp. blue) line corresponds to a poisoned (resp. clean) example. The $x$-value represents the number of backdoored examples present in the training set, while the $y$-value represents the model output (average margin) at that specific example. The rate of change of the model output
represents the feature strength $s_{\phi_p}(k)$.
    We observe that the model output of backdoored images (orange lines) increases as more backdoored examples are included in the training set. In contrast, the model output for clean images (blue lines) is not affected by the number of poisoned training examples.
    }
    \label{fig:margin-lines}
\end{figure}

    \section{A Primitive for Detecting Backdoored Examples}
    \label{sec:framework}
    
The perspective set forth in the previous sections suggests that we need to be able
to analyze the strength of features present in a dataset to understand
the effect of these features on a model's predictions.
Particularly, such an analysis
would allow us to translate \cref{ass:strongest_feature} into an algorithm for
detecting backdoor training examples.
Specifically, we would be able to estimate the feature strength
$s_\feat(k)$ for a given feature $\feat$.
Now, if we had a specific feature $\feat$ in mind,
we could simply compute the feature strength $s_\feat(k)$
using Equations \eqref{eq:k-output} and \eqref{eq:strength} directly.
In our case, however, identifying the feature of interest (i.e., the backdoor feature)
is essentially our goal.

To this end, in this section we first show how to estimate the strength of all
viable features $\feat$ {\em simultaneously}. We then demonstrate how we can
leverage this estimate to detect the strongest one among them.
Our key tool here will be the {\em datamodeling} framework
\citep{ilyas2022datamodels}.
In particular,
\citet{ilyas2022datamodels} have shown that, for every example $z$, and
for a model output function $f$ corresponding to training a deep neural network
and evaluating it on that example,
there exists a weight vector $w_z \in \mathbb{R}^{|S|}$
such that:
\begin{equation}
    \label{eq:approx-margin}
    \mathbb{E}[\modelout(z; S')] \approx \bm{1}_{S'}^\top w_z
\end{equation}
for subsets $S' \sim \mathcal{D}_S$, where $\bm{1}_{S'} \in \{0, 1\}^{|S|}$ is the {\em indicator vector}
of $S'$\footnote{The indicator vector $\bm{1}_{S'}$ takes a value of 1 at index $i$, if training example $z_i \in S'$, and 0 otherwise.}.
In other words, we can approximate the specific outcome of training a deep neural network on a given
subset $S' \subset S$ as a linear function of the presence of each training data example.
As the ability of the datamodeling framework to capture the model output function will be critical to the effectiveness of our method, we state it as an explicit assumption.
\begin{ass}[Datamodel accuracy]
    \label{ass:datamodel-accuracy}
    For any example $z$, with a corresponding datamodel weight $w_z$, we have that
    \begin{align}
        \mathbb{E}_{S' \sim \mathcal{D}_S}\left[
            \left(
                \mathbb{E}[\modelout(z; S')]
                -
                \bm{1}_{S'}^\top w_z
                \right)^2
                \right] \leq \varepsilon
    \end{align}
    where $\epsilon > 0$ represents a bound on the error of estimating the model output function using datamodels.
\end{ass}
\cref{ass:datamodel-accuracy} essentially guarantees that datamodels provide an
accurate estimate of the model output function for any example $z$ and for any random subset $S' \sim \mathcal{D}_S$.
Also, we can in fact verify this assumption by sampling sets $S'$ and
computing the error from \cref{ass:datamodel-accuracy} directly (replacing the inner expectation with an
empirical average).

Now, it turns out that this property alone---captured as a formal lemma
below---suffices to estimate the feature strength $s_\feat(k)$ of any feature $\feat$.

\begin{restatable}{lemma}{lemmaapprox}
\label{lemma:sensitivity-approx}
For a feature $\feat$, let $\bm{1}_{\feat(S)}$ be the indicator vector
of its support $\featsupp(S)$, $\ones$ be the $n$-dimensional vector of ones, and let $h: \mathbb{R}^n \to \mathbb{R}^n$ be
defined as
\[
    h(v) = \frac{1}{\|v\|_1} v - \frac{1}{n - \|v\|_1} (\ones - v).
\]
Then, under \cref{ass:datamodel-accuracy}, we have that
there exists some $C > 0$ such that
\begin{equation}
    \label{eq:approx-sensitivity}
    \left|s_\feat(\alpha \cdot |\featsupp(S)|) - \frac{1}{|\Phi(S)|} \sum_{z \in \Phi(S)} w_z^\top h(\bm{1}_{\feat(S)}) \right|
     \leq \dmerror.
\end{equation}
where $\epsilon$ is as defined in \cref{ass:datamodel-accuracy}.
\end{restatable}
So, \cref{lemma:sensitivity-approx} provides a closed-form expression---involving
only the datamodel weight vectors $\{w_z\}$---for the (approximate) feature strength $s_\phi(k)$
of feature $\feat$. We provide a proof of this lemma in~\cref{app:proof_lem_1}.

\subsection{Poisoned examples as a maximum-sum submatrix}
\label{subsec:submatrix}
In the previous section,
we have shown how we can leverage datamodels
to estimate any given feature's strength.
In this section, we combine \cref{lemma:sensitivity-approx}
and Assumptions \ref{ass:strongest_feature} and \ref{ass:datamodel-accuracy} (i.e., that the backdoor trigger constitutes the
strongest feature in the dataset) into an algorithm that {\em provably} finds
backdoor training examples (provided that Assumptions \ref{ass:strongest_feature} and \ref{ass:datamodel-accuracy} indeed hold).

To this end, recall that $n = |S|$ and $p = |\featsupp_p(S)|$.
\cref{ass:strongest_feature} implies that $s_{\feat_p}(\alpha \cdot p)$
(i.e., the strength of a backdoor feature $\phi_p$) is large.
So, guided by \Cref{lemma:sensitivity-approx}, we consider the following optimization problem:
\begin{align}
    \label{eq:maxsum-sub}
    \arg\max_{v \in \{0,1\}^{n}}
    h(v)^\top \W v
    \qquad \text{s.t.} \qquad \|v_i\|_1 = p,
\end{align}
where $h$ is as defined in \Cref{lemma:sensitivity-approx}.
The following lemma (proved in \cref{app:proof_lem_2}) shows that under \cref{ass:strongest_feature},
the solution to \eqref{eq:maxsum-sub} is precisely the indicator vector of the backdoor trigger feature.
\begin{restatable}{lemma}{lemmamaxim}
    \label{lemma:maximizer-is-indicator}
    Suppose \cref{ass:strongest_feature} holds for some $\delta$ and
     \cref{lemma:sensitivity-approx} holds for some $C$. Then if
    $\delta > 2pC \varepsilon^{1/2} n^{1/4}$,
    the unique maximizer of \eqref{eq:maxsum-sub} is the vector $v_p =\bm{1}_{\feat_p(S)}$,
    i.e., the indicator of the backdoored training examples, where $\epsilon$ is as in \cref{ass:datamodel-accuracy}.
\end{restatable}

Now, the fact that for $v\in\{0, 1\}^n$ we have $\ones^\top \W v = v^\top (\text{diag}(\ones^\top \W)) v$
allows us to express \eqref{eq:maxsum-sub} as a submatrix-sum maximization
problem. In particular, we have that
\begin{align}
    \nonumber
    &\argmax_{v \in \{0,1\}^{n}: \|v\|_1 = p}
    \left( \frac{1}{p}\cdot v - \frac{1}{n - p}\cdot (\ones - v) \right)^\top \W v  \\
    \nonumber
    =&\argmax_{v \in \{0,1\}^{n}: \|v\|_1 = p}
    \left(v^\top \W  - \frac{p}{n} \cdot \ones^\top \W\right) v \\
    =& \argmax_{v \in \{0,1\}^{n}: \|v\|_1 = p}
    v^\top \left(\W - \text{diag}\left(\frac{p}{n}\cdot \ones^\top \W\right)\right) v.
    \label{eq:final_prob}
\end{align}

    \subsection{Detecting backdoored examples}
    \label{subsec:theory_alg}
    The formulation presented in \eqref{eq:maxsum-sub} is difficult to solve
directly, for multiple reasons.
First, the optimization problem requires knowledge of the number of poisoned examples
$|\Phi_p(S)|$, which is unknown in practice.
Second, even if we did know the number of poisoned examples,
the problem is still NP-hard in general \citep{branders2017mining}.
In fact, even linearizing \eqref{eq:maxsum-sub} and using the commercial-grade
mixed-integer linear program solver Gurobi \citep{gurobi} takes several days to solve
(per problem instance) due to the sheer number of optimization variables.

We thus resort to a different approximation.
For each $k$ in a pre-defined set of ``candidate sizes'' $K$ for the submatrix in
\eqref{eq:final_prob}, we set $p$ equal to $k$.
We then solve the resulting maximization problem using a greedy local search algorithm
inspired by the Kernighan-Lin heuristic for graph partitioning \citep{kernighan70}.
That is, starting from a random assignment for $v \in \{0, 1\}^n$,
the algorithm considers all pairs of indices $i, j \in [n]$ such that
$v_i \neq v_j$, and such that swapping the values of $v_i$ and $v_j$
would improve the submatrix sum objective.
The algorithm greedily selects the pair that would most improve the objective
and terminates when no such pair exists.
We run this local search algorithm $T = 1000$ times for each value of $k$ and collect
the {\em candidate solutions} $\{\bm{v}^{k,l}: k \in K, l \in [T]\}$.

Now, rather than using any particular one of these solutions, we combine them to yield
a score for each example $z_{i} \in S$.
We define this score for the $i$-th example $z_i$ as the weighted sum of the number of times
it was included in the solution of local search, that is,
\begin{equation}
    \label{eq:score}
    s_i = \sum_{k \in K} \, \frac{1}{k} \sum_{l=1}^{T} \, \bm{v}^{k,l}_i.
\end{equation}
Intuitively, we expect that backdoored training examples will end up in many
of the greedy local search solutions (due to \cref{ass:strongest_feature}) and thus have a high score $s_i$.
We translate the scores \eqref{eq:score} into a concrete defense by flagging (and removing) the
examples with the highest score.

    \section{Experiments}
    \label{sec:experiments}
    \begin{table*}[!htbp]
    \centering
        \begin{tabular}{@{}ccccc@{}}
            \toprule
            \textbf{Exp.} & \textbf{Attack Type}  & \textbf{Poison ratio} & \textbf{Clean Accuracy} & \textbf{Poisoned Accuracy} \\ \midrule
            1             & Dirty-Label           & 1.5\%               & 86.64               & 19.90    \\
            2             & Dirty-Label           & 5\%                 & 86.67               & 12.92    \\ \midrule
            3             & Dirty-Label           & 1.5\%               & 86.39               & 49.57    \\
            4             & Dirty-Label           & 5\%                 & 86.23               & 10.67    \\ \midrule
            5             & Clean-Label           & 1.5\%               & 86.89               & 75.58    \\
            6             & Clean-Label           & 5\%                 & 87.11               & 41.89    \\ \midrule
            7             & Clean-Label (no adv.) & 5\%                 & 86.94               & 71.68    \\
            8             & Clean-Label (no adv.) & 10\%                & 87.02               & 52.08    \\ \bottomrule
            \end{tabular}
            \caption{
                A summary of the different backdoor attacks we consider.
                Clean-Label (no adv.) is the non-adversarial clean label attack
                from \citet{turner2019label}.
                }
        \label{tab:attack_params}
\end{table*}

In the previous section, we developed an algorithm that provably detects backdoored
examples in a dataset whenever Assumptions~\ref{ass:strongest_feature} and
\ref{ass:datamodel-accuracy} hold.
We now consider several settings\footnote{Table~\ref{tab:attack_params} provides a summary of the poisoning setups.}, and two common types of backdoor attacks: dirty-label
attacks~\citep{gu2017badnets} and clean-label attacks~\citep{turner2019label}.
For each setting, we verify  whether our assumptions hold, and then validate the
effectiveness of our proposed detection algorithm.

\paragraph{Experimental setup.}
In Table~\ref{tab:attack_params}, we present a summary of our experiments. For
all of these experiments, we use the CIFAR-10
dataset~\citep{krizhevsky2009learning}, and the ResNet-9
architecture~\citep{he2015residual}, and compute the datamodels using the framework presented in
\citep{ilyas2022datamodels}. Specifically, for each experiment and setup, we train
a total of 100,000 models, each on a random subset
containing 50\%\footnote{We train, in Exp.~2 from \cref{tab:attack_params}, each model on 30\% of the dataset. More details in Section~\ref{subsec:verify-theory}.} of CIFAR-10\footnote{The chosen value of $\alpha$ from Assumption~\ref{ass:strongest_feature} is hence $1/2$.}, and chosen uniformly at random.

\subsection{Verifying our assumptions}
\label{subsec:verify-theory}

In Section~\ref{sec:theory}, we presented two assumptions required for our proposed defense to (provably) work. We now verify whether these assumptions hold in the experimental settings we consider, then validate the effectiveness of our detection algorithm.

\paragraph{Datamodel accuracy.}
\cref{lemma:sensitivity-approx} states that datamodels are good
approximators of a feature strength (provided \cref{ass:datamodel-accuracy} holds). To validate whether this is indeed the case, we estimate the ``ground-truth''
feature strength $s_{\phi_p}(k)$ of the backdoor triggr feature $\phi_p$ as described in \eqref{eq:k-output} and \eqref{eq:strength}.
More precisely, we train 100,000 models on random subsets of the training set,
each containing 50\% of the training examples. We then compute
how the model outputs change with the inclusion/exclusion of the backdoored examples.
We then compute
the model's outputs for the backdoored examples as a function of the number of backdoored
training examples. We then estimate the feature strength $s_\feat(k)$ as the rate of change
of the model outputs for backdoored examples.

Afterwards, we estimate the feature strength using datamodels, as given by Equation \eqref{eq:approx-sensitivity}.
In particular, we compute the datamodels matrix $\mathbf{W}$, the indicator vector
$h\left(\bm{1}_{\phi_p(S)}\right)$ from Lemma~\ref{lemma:sensitivity-approx},
and their product \[
     h\left(\bm{1}_{\phi_p(S)}\right)^\top \mathbf{W} \in \mathbb{R}^{|S|}
    \]
Each entry of this product is an estimate of the backdoor feature strength at every
training example. As Figure~\ref{fig:theory-sensitivity} shows, datamodels are indeed good
approximators of the backdoor trigger feature's strength.

\begin{figure}[!h]
    \centering
    \includegraphics[width=.5\linewidth]{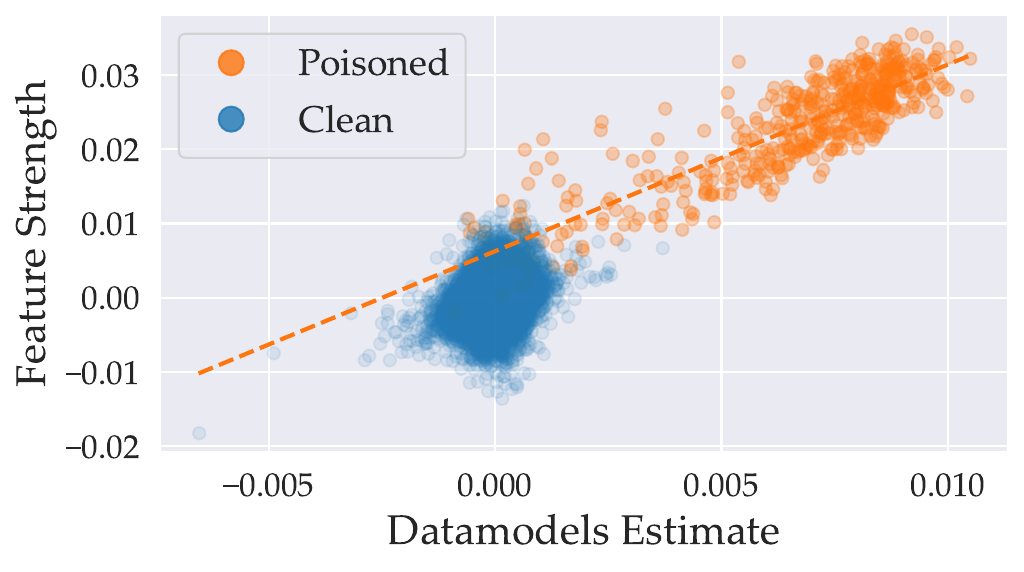}

    \caption{
    {\bf Estimating feature strength using datamodels.}
    Each orange (resp. blue) data point in the scatter plot above represents a poisoned
    (resp. clean) training example. The $x$-value of each data point represents the
    feature strength estimated using datamodels (see Equation \eqref{eq:approx-sensitivity}),
    and the $y$-value represents the feature strength as estimated using Equation \eqref{eq:strength}.
    We see a strong linear correlation between these two quantities for poisoned
    examples, which indicates that datamodels provide a good estimate of feature strength.
 }
    \label{fig:theory-sensitivity}
\end{figure}

\paragraph{Backdoor trigger as the strongest feature.}
Recall that assumption~\ref{ass:strongest_feature} states that the backdoor trigger feature has the highest feature strength among all the features present in the dataset. To validate this assumption in our settings, we leverage our approximation of feature strength using datamodels. Specifically, our result from Section~\ref{sec:framework} suggest;s that the obtained product should be highly correlated with the ground-truth backdoor trigger indicator vector $\bm{1}_{\feat_p(S)}$. We thus measure this correlation by computing the area under the ROC curve (AUROC) between the product $h(\bm{1}_{\feat_p(S)})^\top \mathbf{W}$ and the indicator vector $\bm{1}_{\feat_p(S)}$. As we can see in Table~\ref{tab:gt-auroc}, the AUROC score is very high in seven out of the eight settings, which suggests that \cref{ass:strongest_feature} indeed holds in these cases.

Interestingly, we observe that we get a low AUROC in the second experiment from \cref{tab:attack_params} (the one with a very large number of backdoored examples), which indicates that one of our assumptions does not hold in that case. To investigate the reason for that, we inspect the backdoor feature strength. \cref{fig:margin-lines-exp-2} shows that, for subsets of the training set containing 50\% of the training examples, the model output does not change as the number of poisoned samples increases, i.e., for these subsets, the backdoor feature strength is essentially 0. Consequently, \cref{ass:strongest_feature} indeed does not hold.
To fix this problem, we use smaller random subsets, i.e., ones containing 30\% of the training examples, when estimating the datamodels. In this new setting, the backdoor feature strength is significantly higher, and the AUROC between the poison indicator vector and the feature strenght product jumps to 99.34\%. For the remainder of the paper, we will use this setup.

\begin{figure}[!h]
    \centering
    \includegraphics[width=.55\linewidth]{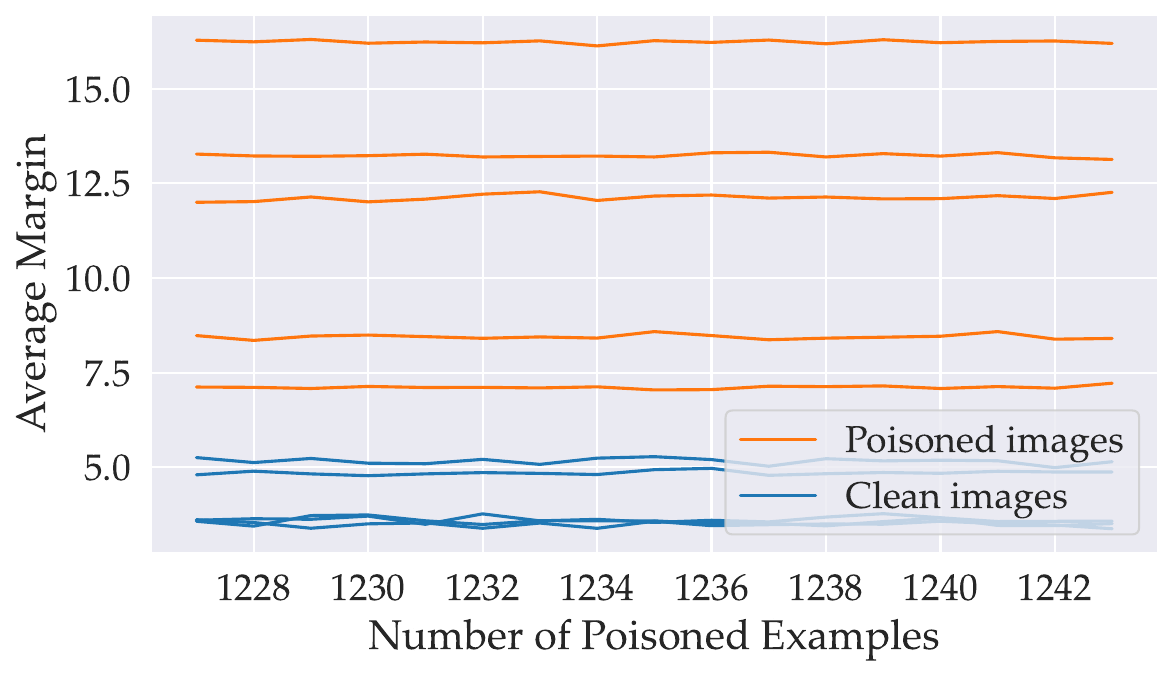}
\caption{
    {\bf Model output for different number of backdoor training examples}.
Each orange (resp. blue) line corresponds to a poisoned (resp. clean) example.
The $x$-value represents the number of backdoored examples present in the training set,
while the $y$-value represents the model output (average margin) at that specific example.
The rate of change of the model output
represents the feature strength. We observe that for backdoored examples (orange lines) from Exp. 2 (see \cref{tab:attack_params}), the model output does not
change as more training examples are poisoned. Consequently, the
backdoor feature strength is 0.
}
\label{fig:margin-lines-exp-2}
\end{figure}

\begin{table}[!h]
    \centering
    \begin{tabular}{@{}cccccccc@{}}
        \toprule
        {\bf E1} & {\bf E2} & {\bf E3} & {\bf E4} & {\bf E5} & {\bf E6} & {\bf E7} & {\bf E8} \\ \midrule
        99.9 & 60.9 & 98.0 & 97.7 & 99.9 & 99.9 & 97.0 & 98.3 \\ \bottomrule
    \end{tabular}
    \caption{
        AUROC of the backdoor feature strength and the backdoor examples indicator vector for our
        setups from \cref{tab:attack_params}.
        }
    \label{tab:gt-auroc}
\end{table}

\begin{table*}[!t]
    \begin{minipage}[t]{\linewidth}
        \setlength{\tabcolsep}{6pt}
        \renewcommand{\arraystretch}{1.2}
        \resizebox{\linewidth}{!}{
            \begin{tabular}{@{}ccccccccccccc@{}}
                \toprule
                \multirow{2}{*}{\textbf{Exp.}} & \multicolumn{2}{c}{\textbf{No Defense}}                 & \multicolumn{2}{c}{\textbf{AC}}                         & \multicolumn{2}{c}{\textbf{ISPL}}                       & \multicolumn{2}{c}{\textbf{SPECTRE}}                    & \multicolumn{2}{c}{\textbf{SS}}                         & \multicolumn{2}{c}{\textbf{Ours}}  \\ \cline{2-13}
                & \textit{Clean} & \textit{Poisoned} & \textit{Clean} & \textit{Poisoned} & \textit{Clean} & \textit{Poisoned} & \textit{Clean} & \textit{Poisoned} & \textit{Clean} & \textit{Poisoned} & \textit{Clean} & \textit{Poisoned} \\ \midrule
                1              & 86.64          & 19.90            & 86.76          & 19.68            & 86.13          & {\bf86.15}            & 86.71          & 20.17            & 85.52          & 30.99            & 85.05          & {\bf 85.06}            \\
                2              & 86.67          & 12.92            & 85.41          & 12.93            & 85.88          & {\bf85.82}            & -              & -                & 85.33          & 13.63            & 83.39          & {\bf83.13}             \\ \midrule
                3              & 86.39          & 49.57            & 86.25          & 48.85            & 86.32          & {\bf85.57}             & 86.28          & 45.32            & 85.22          & 78.22            & 84.82         & {\bf 84.11}            \\
                4              & 86.23          & 10.67            & 84.75          & 10.82            & 85.86           & {\bf85.18}            & -              & -                & 84.85           & 13.33            & 84.64          & {\bf 83.72}            \\ \midrule
                5              & 86.89          & 75.58            & 86.73          & 82.83            & 86.04          & {\bf85.89}            & 86.82          & 80.65            & 85.67          & {\bf85.41}            & 83.82          & {\bf 83.72}            \\
                6              & 87.11          & 41.89            & 86.85          & 51.05            & 86.18          & {\bf86.11}            & 86.97          & 51.18            & 85.68         & {\bf85.60}            & 84.88          & {\bf 84.79}            \\ \midrule
                7              & 87.02          & 71.68            & 86.90          & 73.28            & 86.50          & 82.31            & 86.72          & 76.97            & 85.70          & 82.70            & 84.19         & {\bf 84.02}            \\
                8              & 86.94          & 52.08            & 86.81          & 56.78            & 86.04          & 71.27            & 86.63          & 52.27            & 85.87          & 71.93            & 84.81          & {\bf 84.66}             \\            \bottomrule
            \end{tabular}
            }
            \caption{
                A summary of the model performances on a ``clean'' and ``poisoned'' validation sets after applying our method, as well as several baselines in the settings we consider. The high accuracy on both the clean and poisoned validation
                sets indicates the effectiveness of our defense against the considered backdoor attacks.
                }
            \label{tab:results_acc}
        \vspace{0.2in}

        \end{minipage}
        \label{tab:results_all}
\end{table*}

\subsection{Evaluating the effectiveness of our backdoor defense}
Let us evaluate the effectiveness of our method (see \cref{subsec:theory_alg}) in defending against backdoor attacks.

\paragraph{Evaluating our score.}
As a first step, we measure how well our scores predict the backdoored examples in our eight settings. Specifically, we compute our scores by running our local search algorithm from Section~\ref{subsec:theory_alg} on the datamodels matrix $\mathbf{W}$, then aggregating the results from the different runs. Following that, we check how well these scores correlate with the backdoor examples indicator vector $\bm{1}_{\feat_p(S)}$.  As Table~\ref{tab:alg-auroc} shows, there is a high correlation between these two quantities in all setups (cf. Section~\ref{subsec:verify-theory})\footnote{For Exp.~2 from \cref{tab:attack_params}, we are using the datamodels computed with $\alpha=30\%$, as described in Section~\ref{subsec:verify-theory}}. This high correlation suggests that our local search algorithm generates a score that is predictive of the backdoored examples.

\begin{table}[h]
    \centering
    \begin{tabular}{@{}cccccccc@{}}
        \toprule
        {\bf E1} & {\bf E2} & {\bf E3} & {\bf E4} & {\bf E5} & {\bf E6} & {\bf E7} & {\bf E8} \\ \midrule
        94.3 & 92.25 & 74.4 & 80.2 & 93.4 & 93.2 & 91.1 & 95.5 \\ \bottomrule
    \end{tabular}
    \caption{
        AUROC for our scores (see \cref{subsec:theory_alg}) and the backdoor indicator vector
        for our setups from \cref{tab:attack_params}.
        }
    \label{tab:alg-auroc}
\end{table}

\paragraph{Evaluating the effectiveness of our proposed defense.}
Given that our scores are predictive of the backdoor examples indicator vector,
we expect that removing the examples with the highest scores will be an effective
defense against the backdoor attack. To test this claim, for each of the backdoor attacks settings, we train a model on the backdoored dataset, and compute the accuracy of
this model on (a) the clean validation set, (b) and on the backdoored validation
set\footnote{By adding the trigger to all images of the clean validation set.}.
We then remove from the training set the examples corresponding to the
top 10\% of the scores\footnote{
    We remove 20\% in Exp.~2 from \cref{tab:attack_params} since the number of poisoned examples is larger.
},
train a new model on the resulting dataset, and then check the performance
of this new model on the clean and the fully-backdoored validation sets.
We also compare our detection algorithm with several baselines, including
Inverse Self-Paced Learning\footnote{We thank the authors for sharing their
implementation.} (ISPL)~\citep{jin2021poisonselfexpansion}, Spectral
Signatures\footnote{We re-implement the algorithm presented in the paper.}
(SS)~\citep{tran2018spectral},
SPECTRE\footnote{\href{https://github.com/SewoongLab/spectre-defense}{https://github.com/SewoongLab/spectre-defense}}~\citep{hayase2021spectre}
and Activation Clustering (AC)~\citep{chen2018detecting}.
\cref{tab:results_acc} shows that there is no substantial
drop in accuracy when evaluating the models trained on the training set curated using our approach.

    \section{Related Work}
    \label{sec:related_work}
    
Developing backdoor attacks and defenses in the context of deep learning is a
very active area of research~\citep{
gu2017badnets, tran2018spectral, chen2018detecting, turner2019label, saha2020hidden,
shokri2020bypassing, hayase2021spectre, qi2022circumventing, goldblum2022dataset,
goldwasser2022planting}
---see e.g.~\citep{li2022backdoor}
for a survey.
One popular approach to defending against backdoor attacks revolves around outlier detection
in the latent space of neural networks~\citep{tran2018spectral, chen2018detecting,
hayase2021spectre}. However, such defenses inherently fail in defending against adaptive attacks that
leave no trace in the latent space of backdoored models~\citep{shokri2020bypassing}.

Another line of work investigates certified defenses against backdoor
attacks~\citep{levine2021deeppartition, wang2022certifiedpoison}. To accomplish that, the proposed methods provide
certificates by training separate models on different partitions of the training set, and dropping the
models trained on data containing backdoored examples. This approach, however, significantly degrades the accuracy of the
trained model, and is only able to  certify accuracy when the number of backdoored examples is very small.

A number of prior works explore the applicability of influence-based
methods as defenses against different attacks in deep
learning~\citep{koh2017understanding}. To the best of our
knowledge, only \citet{lin2022measuring} discuss using such methods for
defending against backdoor attacks.
However, their defense requires knowledge of the attack
parameters that are typically unknown.
Closest to our work is that of~\citet{jin2021poisonselfexpansion},
who (similar to this work) consider a defense based on model behavior rather than
properties of any latent space.

    \section{Conclusion}
    \label{sec:conclusion}
    In this paper, we presented a new perspective on backdoor attacks. Specifically,
we argue that backdoor triggers are fundamentally
indistinguishable from existing features in the dataset.
Consequently, the task of detecting backdoored training examples becomes
equivalent to that of detecting strong features.
Based on this observation, we propose a primitive---and a corresponding algorithm---for
identifying and removing backdoored examples.
Through a wide range of backdoor attacks, we demonstrated the effectiveness of
our approach and its ability to retain high
accuracy of the resulting models.

    \section{Acknowledgments}
    \label{sec:acks}
    Work supported in part by the NSF grants CNS-1815221 and DMS-2134108, and Open Philanthropy. This material is based upon work supported by the Defense Advanced Research Projects Agency (DARPA) under Contract No. HR001120C0015.
We thank the MIT Supercloud cluster~\citep{reuther2018interactive} for providing computational resources that supported part of this work.

    \clearpage
    \printbibliography

@inproceedings{reuther2018interactive,
  title        = {Interactive supercomputing on 40,000 cores for machine learning and data analysis},
  author       = {Reuther, Albert and Kepner, Jeremy and Byun, Chansup and Samsi, Siddharth and Arcand, William and Bestor, David and Bergeron, Bill and Gadepally, Vijay and Houle, Michael and Hubbell, Matthew and Jones, Michael and Klein, Anna and Milechin, Lauren and Mullen, Julia and Prout, Andrew and Rosa, Antonio and Yee, Charles and Michaleas, Peter},
  booktitle    = {2018 IEEE High Performance extreme Computing Conference (HPEC)},
  pages        = {1--6},
  year         = {2018},
  organization = {IEEE}
}

@article{adi2018turning,
  title   = {Turning Your Weakness Into a Strength: Watermarking Deep Neural Networks by Backdooring},
  author  = {Adi, Yossi and Baum, Carsten and Cisse, Moustapha and Pinkas, Benny and Keshet, Joseph},
  journal = {$\{$USENIX$\}$ Security Symposium},
  year    = {2018}
}

@inproceedings{branders2017mining,
  title     = {Mining a Sub-Matrix of Maximal Sum},
  author    = {Vincent Branders and Pierre Schaus and Pierre Dupont},
  booktitle = {International Workshop on New Frontiers in Mining Complex Patterns (NFMCP)},
  year      = {2017}
}

@article{chen2017targeted,
  title   = {Targeted Backdoor Attacks on Deep Learning Systems Using Data Poisoning},
  author  = {Chen, Xinyun and Liu, Chang and Li, Bo and Lu, Kimberly and Song, Dawn},
  journal = {arXiv preprint arXiv:1712.05526},
  year    = {2017}
}

@article{chen2018detecting,
  title   = {Detecting backdoor attacks on deep neural networks by activation clustering},
  author  = {Chen, Bryant and Carvalho, Wilka and Baracaldo, Nathalie and Ludwig, Heiko and Edwards, Benjamin and Lee, Taesung and Molloy, Ian and Srivastava, Biplav},
  journal = {arXiv preprint arXiv:1811.03728},
  year    = {2018}
}

@article{goldblum2022dataset,
  title     = {Dataset security for machine learning: Data poisoning, backdoor attacks, and defenses},
  author    = {Goldblum, Micah and Tsipras, Dimitris and Xie, Chulin and Chen, Xinyun and Schwarzschild, Avi and Song, Dawn and Madry, Aleksander and Li, Bo and Goldstein, Tom},
  journal   = {IEEE Transactions on Pattern Analysis and Machine Intelligence},
  year      = {2022},
  publisher = {IEEE}
}

@article{goldwasser2022planting,
  title   = {Planting Undetectable Backdoors in Machine Learning Models},
  author  = {Goldwasser, Shafi and Kim, Michael P and Vaikuntanathan, Vinod and Zamir, Or},
  journal = {arXiv preprint arXiv:2204.06974},
  year    = {2022}
}

@article{gu2017badnets,
  title   = {Badnets: Identifying Vulnerabilities in the Machine Learning Model Supply Chain},
  author  = {Gu, Tianyu and Dolan-Gavitt, Brendan and Garg, Siddharth},
  journal = {arXiv preprint arXiv:1708.06733},
  year    = {2017}
}

@article{hayase2021spectre,
  title   = {SPECTRE: defending against backdoor attacks using robust statistics},
  author  = {Hayase, Jonathan and Kong, Weihao and Somani, Raghav and Oh, Sewoong},
  journal = {arXiv preprint arXiv:2104.11315},
  year    = {2021}
}

@inproceedings{ho2020ddpm,
  title     = {Denoising Diffusion Probabilistic Models},
  author    = {Ho, Jonathan and Jain, Ajay and Abbeel, Pieter},
  booktitle = {Neural Information Processing Systems (NeurIPS)},
  year      = {2020}
}

@inproceedings{ilyas2022datamodels,
  title     = {Datamodels: Predicting Predictions from Training Data},
  author    = {Ilyas, Andrew and Park, Sung Min and Engstrom, Logan and Leclerc, Guillaume and Madry, Aleksander},
  booktitle = {International Conference on Machine Learning (ICML)},
  year      = {2022}
}

@article{jin2021poisonselfexpansion,
  author    = {Jin, Charles and Sun, Melinda and Rinard, Martin},
  year      = {2021},
  title     = {Provable Guarantees against Data Poisoning Using Self-Expansion and Compatibility},
  booktitle = {ArXiv preprint arxiv:2105.03692}
}

@inproceedings{koh2017understanding,
  title     = {Understanding Black-box Predictions via Influence Functions},
  author    = {Koh, Pang Wei and Liang, Percy},
  booktitle = {International Conference on Machine Learning},
  year      = {2017}
}

@inproceedings{krizhevsky2009learning,
  title     = {Learning Multiple Layers of Features from Tiny Images},
  author    = {Krizhevsky, Alex},
  booktitle = {Technical report},
  year      = {2009}
}

@inproceedings{levine2021deeppartition,
  title     = {Deep Partition Aggregation: Provable Defenses against General Poisoning Attacks},
  author    = {Alexander Levine and Soheil Feizi},
  booktitle = {International Conference on Learning Representations},
  year      = {2021}
}

@article{li2022backdoor,
  title     = {Backdoor learning: A survey},
  author    = {Li, Yiming and Jiang, Yong and Li, Zhifeng and Xia, Shu-Tao},
  journal   = {IEEE Transactions on Neural Networks and Learning Systems},
  year      = {2022},
  publisher = {IEEE}
}

@article{lin2022measuring,
  title   = {Measuring the Effect of Training Data on Deep Learning Predictions via Randomized Experiments},
  author  = {Lin, Jinkun and Zhang, Anqi and Lecuyer, Mathias and Li, Jinyang and Panda, Aurojit and Sen, Siddhartha},
  journal = {arXiv preprint arXiv:2206.10013},
  year    = {2022}
}

@article{qi2022circumventing,
  title   = {Circumventing Backdoor Defenses That Are Based on Latent Separability},
  author  = {Qi, Xiangyu and Xie, Tinghao and Mahloujifar, Saeed and Mittal, Prateek},
  journal = {arXiv preprint arXiv:2205.13613},
  year    = {2022}
}

@inproceedings{saha2020hidden,
  title     = {Hidden trigger backdoor attacks},
  author    = {Saha, Aniruddha and Subramanya, Akshayvarun and Pirsiavash, Hamed},
  booktitle = {Proceedings of the AAAI conference on artificial intelligence},
  volume    = {34},
  number    = {07},
  pages     = {11957--11965},
  year      = {2020}
}

@inproceedings{shafahi2018poison,
  title     = {Poison frogs! targeted clean-label poisoning attacks on neural networks},
  author    = {Shafahi, Ali and Huang, W Ronny and Najibi, Mahyar and Suciu, Octavian and Studer, Christoph and Dumitras, Tudor and Goldstein, Tom},
  booktitle = {Advances in Neural Information Processing Systems (NeurIPS)},
  year      = {2018}
}

@inproceedings{shokri2020bypassing,
  title        = {Bypassing backdoor detection algorithms in deep learning},
  author       = {Shokri, Reza and others},
  booktitle    = {2020 IEEE European Symposium on Security and Privacy (EuroS\&P)},
  pages        = {175--183},
  year         = {2020},
  organization = {IEEE}
}

@inproceedings{song2019diffusion,
  title     = {Generative Modeling by Estimating Gradients of the Data Distribution},
  author    = {Song, Yang and Ermon, Stefano},
  booktitle = {Neural Information Processing Systems (NeurIPS)},
  year      = {2019}
}

@inproceedings{tran2018spectral,
  title     = {Spectral Signatures in Backdoor Attacks},
  author    = {Tran, Brandon and Li, Jerry and M{\k{a}}dry, Aleksander},
  booktitle = {Advances in Neural Information Processing Systems (NeurIPS)},
  year      = {2018}
}

@inproceedings{turner2019label,
  title  = {Label-Consistent Backdoor Attacks},
  author = {Turner, Alexander and Tsipras, Dimitris and Madry, Aleksander},
  year   = {2019}
}

@inproceedings{wang2019neural,
  title     = {Neural cleanse: Identifying and mitigating backdoor attacks in neural networks},
  author    = {Wang, Bolun and Yao, Yuanshun and Shan, Shawn and Li, Huiying and Viswanath, Bimal and Zheng, Haitao and Zhao, Ben Y},
  booktitle = {Proceedings of 40th IEEE Symposium on Security and Privacy},
  year      = {2019}
}

@inproceedings{wang2022certifiedpoison,
  author    = {Wang, Wenxiao and Levine, Alexander and Feizi, Soheil},
  title     = {Improved Certified Defenses against Data Poisoning with (Deterministic) Finite Aggregation},
  booktitle = {International Conference on Machine Learning},
  year      = {2022}
}

@inproceedings{xie2020dba,
  title     = {DBA: Distributed Backdoor Attacks against Federated Learning},
  author    = {Chulin Xie and Keli Huang and Pin-Yu Chen and Bo Li},
  booktitle = {International Conference on Learning Representations},
  year      = {2020}
}

@inproceedings{yang2022notallpoisonequal,
  title     = {Not All Poisons are Created Equal: Robust Training against Data Poisoning},
  author    = {Yang, Yu and Liu, Tian Yu and Mirzasoleiman, Baharan},
  booktitle = {International Conference on Machine Learning},
  year      = {2022}
}

@misc{he2015residual,
  author    = {He, Kaiming and Zhang, Xiangyu and Ren, Shaoqing and Sun, Jian},
  booktitle = {Conference on Computer Vision and Pattern Recognition (CVPR)},
  title     = {Deep Residual Learning for Image Recognition},
  year      = 2015
}

@article{kernighan70,
  author  = {Kernighan, B. W. and Lin, S.},
  journal = {The Bell System Technical Journal},
  title   = {An efficient heuristic procedure for partitioning graphs},
  year    = {1970}
}

@inproceedings{feldman2019does,
  author    = {Vitaly Feldman},
  booktitle = {Symposium on Theory of Computing (STOC)},
  title     = {Does Learning Require Memorization? A Short Tale about a Long Tail},
  year      = {2019}
}

@inproceedings{leclerc20213db,
  title     = {3DB: A Framework for Debugging Computer Vision Models},
  author    = {Leclerc, Guillaume and Salman, Hadi and Ilyas, Andrew and Vemprala, Sai and Engstrom, Logan and Vineet, Vibhav and Xiao, Kai and Zhang, Pengchuan and Santurkar, Shibani and Yang, Greg and others},
  booktitle = {arXiv preprint arXiv:2106.03805},
  year      = {2021}
}

@book{hampel2011robust,
  title     = {Robust statistics: the approach based on influence functions},
  author    = {Hampel, Frank R and Ronchetti, Elvezio M and Rousseeuw, Peter J and Stahel, Werner A},
  volume    = {196},
  year      = {2011},
  publisher = {John Wiley \& Sons}
}

@inproceedings{jia2021intrinsic,
  title     = {Intrinsic Certified Robustness of Bagging against Data Poisoning Attacks},
  author    = {Jinyuan Jia and Xiaoyu Cao and Neil Zhenqiang Gong},
  booktitle = {AAAI},
  year      = {2021}
}

@misc{gurobi,
  author = {{Gurobi Optimization, LLC}},
  title  = {{Gurobi Optimizer Reference Manual}},
  year   = {2021},
  url    = {https://www.gurobi.com}
}

@inproceedings{devries2017cutout,
  title     = {Improved Regularization of Convolutional Neural Networks with Cutout},
  author    = {DeVries, Terrance and Taylor, Graham W},
  booktitle = {arXiv preprint arXiv:1708.04552},
  year      = {2017}
}

@misc{leclerc2022ffcv,
  author       = {Guillaume Leclerc and Andrew Ilyas and Logan Engstrom and Sung Min Park and Hadi Salman and Aleksander Madry},
  title        = {ffcv},
  year         = {2022},
  howpublished = {\url{https://github.com/libffcv/ffcv/}}
}

@article{lugosi2019sub,
  title     = {Sub-Gaussian estimators of the mean of a random vector},
  author    = {Lugosi, G{\'a}bor and Mendelson, Shahar},
  journal   = {The annals of statistics},
  volume    = {47},
  number    = {2},
  pages     = {783--794},
  year      = {2019},
  publisher = {Institute of Mathematical Statistics}
}

@inproceedings{liu19abs,
  author    = {Liu, Yingqi and Lee, Wen-Chuan and Tao, Guanhong and Ma, Shiqing and Aafer, Yousra and Zhang, Xiangyu},
  title     = {ABS: Scanning Neural Networks for Back-Doors by Artificial Brain Stimulation},
  year      = {2019},
  booktitle = {ACM SIGSAC Conference on Computer and Communications Security}
}

@inproceedings{carlini2023poisoningwebscale,
      title={Poisoning Web-Scale Training Datasets is Practical},
      author={Carlini, Nicholas and Jagielski, Matthew and Choquette-Choo, Christopher A. and Paleka, Daniel and Pearce, Will and Anderson, Hyrum and Terzis, Andreas and Thomas, Kurt and Tramèr, Florian},
      year={2023},
      booktitle={arXiv preprint arXiv:2302.10149},
}

@inproceedings{huang22decoupling,
  title     = {Backdoor Defense via Decoupling the Training Process},
  author    = {Huang, Kunzhe and Li, Yiming and Wu, Baoyuan and Qin, Zhan and Ren, Kui},
  booktitle = {International Conference on Learning Representations (ICLR)},
  year      = {2022}
}

@inproceedings{zeng21poisonfreq,
author = {Zeng, Yi and Park, Won and Mao, Z. Morley and Jia, Ruoxi},
booktitle = {International Conference on Computer Vision (ICCV)},
title = {Rethinking the Backdoor Attacks Triggers: A Frequency Perspective},
year = {2021},
}

@inproceedings{liu22friendlynoise,
      title={Friendly Noise against Adversarial Noise: A Powerful Defense against Data Poisoning Attacks},
      author={Liu, Tian Yu and Yang, Yu and Mirzasoleiman, Baharan},
      year={2022},
      booktitle={arXiv preprint arXiv:2208.10224},
}

    \clearpage

    \appendix

    \addcontentsline{toc}{section}{Appendix} %
    \part{Appendix} %
    \parttoc %

    \onecolumn
    \section{Model Predictions Formulation Used in Our Paper}
    \label{app:app_margin}
    In this work, we use the {\em margin function} (defined below) as the {\em model output function} $\modelout(z; S)$.
\begin{definition}[Margin function]
    For a dataset $S' \subset S$ and a fixed $z=(x, y) \in \mathcal{X} \times \mathcal{Y}$, the
    {\em margin function} $\modelout(x; S)$ is defined as
    $$\modelout(z; S') := \text{the correct-class margin on $z$ of a model trained on S'},$$
    where the correct-class margin is the logit of the correct class minus the largest
    incorrect logit.
\end{definition}
Intuitively, $\modelout(z; S')$ maps from an example $z$ and any
subset of the training dataset $S' \subset S$ to the correct-class margin on $z$
after training (using any fixed learning algorithm) on $S'$.

Here we focus on margins because of their (empirically observed) suitability for
ordinary least squares, as observed in~\citep[Appendix C]{ilyas2022datamodels}.
    \newpage

    \section{Proof of \cref{lemma:sensitivity-approx}}
    \label{app:proof_lem_1}
    \lemmaapprox*

We decompose the proof into two parts. In the first part, we show that we can
approximate the sensitivity $s_\phi(k)$ using datamodel weights
(\cref{lemma:datamodel-approximation}). In the second part, we relate $h$ to the
expressions involving datamodel weights in \cref{lemma:datamodel-approximation}
(\cref{lemma:datamodel-expectation-calculation}).

\begin{lemma}
\label{lemma:datamodel-approximation}
Let $\alpha \in (0,1)$, and $S'$ be a subset of $S$, sampled uniformly at
random, such that $|S'| = \alpha \cdot |S|$. Let $a:=\alpha\cdot|\featsupp(S)|$.
Suppose $\alpha$ is such that $c \leq \alpha \leq 1 - c$ for some absolute
constant $c\in (0, 1)$. Then, there exists a constant $C>0$ (dependending on
$c$) such that we have:
\begin{align} \label{eq:sens-diff-exp}
    & \left|
    s_\feat(a)
    - \E_{z\sim\featsupp(S)}\left[
     \E_{S' \sim \mathcal{D}_S}
    \left[w_z^\top \bm{1}_{S'}\,\bigg|\,|\featsupp(S')| = a + 1\right]
    + \E_{S' \sim \mathcal{D}_S}
    \left[w_z^\top \bm{1}_{S'}\,\bigg|\,|\featsupp(S')| = a\right]
    \right]
    \right| \leq \dmerror.
\end{align}
\end{lemma}
\begin{proof}
\label{proof:lemma_dm_approx}

Recall that the sensitivity $s_\phi(k)$ is defined as:
\begin{align}
    s_{\feat}(k)  :=& g_{\feat}(k+1) - g_\feat(k)\\
    =& \E_{z\sim\featsupp(S)}\left[
     \E_{S' \sim \mathcal{D}_S, S'\not\ni z}
    \left[f(z;S')\,\bigg|\,|\featsupp(S')| = k+1\right]\nonumber
    -
    \E_{S' \sim \mathcal{D}_S, S'\not\ni z}
    \left[f(z;S')\,\bigg|\,|\featsupp(S')| = k\right]
    \right]
\end{align}

For convenience, assume that $a=\alpha \cdot |\featsupp(S)|$ is an integer.
First, by triangle inequality, it is enough to show that:
\begin{equation}
\label{eq:expec_lowerbound_1}
\begin{split}
&\left|
\E_{z\sim\featsupp(S)} \left[
\E_{S' \sim \mathcal{D}_S}
\left[\modelout(z; S')\,\bigg|\,|\featsupp(S')| = a\right]
-
\E_{S' \sim \mathcal{D}_S}
\left[w_z^\top \bm{1}_{S'}\,\bigg|\,|\featsupp(S')| = a\right]
\right]
\right| \leq \frac{1}{2}\cdot\dmerror
\end{split}
\end{equation}
and
\begin{equation}
    \label{eq:expec_lowerbound_2}
\begin{split}
&\left|
\E_{z\sim\featsupp(S)} \left[
\E_{S' \sim \mathcal{D}_S}
\left[\modelout(z; S')\,\bigg|\,|\featsupp(S')| = a + 1\right]
-
\E_{S' \sim \mathcal{D}_S}
\left[w_z^\top \bm{1}_{S'}\,\bigg|\,|\featsupp(S')| = a + 1\right]
\right]
\right| \leq \frac{1}{2} \cdot \dmerror.
\end{split}
\end{equation}

We address~\cref{eq:expec_lowerbound_1}, and the bound
for~\cref{eq:expec_lowerbound_2} follows analogously.

To address \cref{eq:expec_lowerbound_1}, we will show a stronger (per-example)
statement:
\begin{equation}
\label{eq:without_z}
\begin{split}
&\left|
\E_{S' \sim \mathcal{D}_S}
\left[\modelout(z; S')\,\bigg|\,|\featsupp(S')| = a\right]
-
\E_{S' \sim \mathcal{D}_S}
\left[w_z^\top \bm{1}_{S'}\,\bigg|\,|\featsupp(S')| = a\right]
\right| \leq \frac{1}{2}\cdot\dmerror
\end{split}
\end{equation}
Trivially, this implies that the expectation over $z$ is also bounded from above by $\frac{1}{2}\cdot\dmerror$.
To show that \cref{eq:without_z} holds, we first compute the probability of a
random subset $S'$ containing $a$ poisoned samples, and then show an upper bound
for~\cref{eq:expec_lowerbound_1} leveraging the derived probability by using the
definition of conditional expectation.

By directly counting, we have that
\begin{align*}
    \Pr_{S'\sim \mathcal{D}_S}
    {\left[|\featsupp(S')| = a\right]} =
    \frac{
            {\,|\featsupp(S)|\choose a}
            \binom{\,|S|-|\featsupp(S)|}{\alpha\cdot |S|-a}
        }
        {
            \binom{\,|S|}{\alpha \cdot |S|}
        }.
\end{align*}
To ease notation, let $n:=|S|$ and $p:=|\featsupp(S)|$, thus $a=\alpha p$.
Rewriting, we have
\begin{align*}
    \Pr_{S'\sim \mathcal{D}_S}
    {\left[|\featsupp(S')| = \alpha p \right]} =
    \frac{
            \binom{p}{\alpha p}
            \binom{\,n-p}{\alpha (n - p)}
        }
        {
            \binom{\,n}{\alpha n}
        }.
\end{align*}

Next,
\begin{align*}
    &\Pr_{S'\sim \mathcal{D}_S}{\left[|\featsupp(S')| = \alpha p + 1\right]} =
    \frac{
        {p\choose \alpha p + 1}{n-p\choose \alpha(n-p) - 1}
        }
        {{n\choose \alpha n}}\\
    &= \left(\frac{p(1-\alpha)}{\alpha p+1}\cdot\frac{\alpha(n-p)}{(1 - \alpha)(n-p) + 1}\right) \cdot
    \Pr_{S'\sim \mathcal{D}_S}{\left[|\featsupp(S')| = \alpha p\right]}
\end{align*}
We first show that the ratio of the two probabilities is bounded by a constant, i.e.,
\begin{align*}
    &\frac{p(1-\alpha)}{\alpha p+1}\cdot\frac{\alpha(n-p)}{(1 - \alpha)(n-p) + 1} \\
    &= \frac{\alpha}{\alpha + \frac{1}{p}} \cdot \frac{1-\alpha}{1-\alpha + \frac{1}{n-p}}
    \\
    &\geq \frac{\alpha}{\alpha + \alpha} \cdot \frac{1-\alpha}{2-\alpha}
    \geq \frac{1}{2} \cdot \frac{c}{2-c}
\end{align*}
where we used that $1\leq \alpha p$ and $c \leq \alpha \leq 1
-c$.
Thus
\begin{align*}
    \Pr_{S'\sim \mathcal{D}_S}{\left[|\featsupp(S')| = \alpha p + 1\right]}
    \geq
    \frac{c}{2(2-c)}
    \Pr_{S'\sim \mathcal{D}_S}{\left[|\featsupp(S')| = \alpha p\right]}.
\end{align*}
Now we proceed with bounding $\Pr_{S'\sim \mathcal{D}_S}{\left[|\featsupp(S')| =
\alpha p\right]}$. Using Stirling's approximation, we have
\begin{align*}
&\Pr_{S'\sim \mathcal{D}_S}{\left[|\featsupp(S')| = \alpha p\right]} \asymp
\sqrt{\frac{n}{p(n-p)}\frac{1}{\alpha(1-\alpha)}} \geq \frac{2}{C^2\cdot \sqrt{n}}
\end{align*}
for some constant $C>0$. Now from the triangle inequality, Jensen's inequality and
Markov's inequality we have that for sufficiently large $n$
\begin{align*}
&\mathbb{E}_{S' \sim \mathcal{D}_S}\left[\modelout(z; S')\,\bigg|\,|\featsupp(S')| = \alpha p\right]
- \mathbb{E}_{S' \sim \mathcal{D}_S}\left[w_z^\top \bm{1}_{S'}\,\bigg|\,|\featsupp(S')| = \alpha p\right] \\
&\leq \mathbb{E}_{S' \sim \mathcal{D}_S}\left[\left|\modelout(z; S') - w_z^\top \bm{1}_{S'}\right|\,\bigg|\,|\featsupp(S')| = \alpha p\right] \\
&\leq \sqrt{
    \mathbb{E}_{S' \sim \mathcal{D}_S}
    \left[
            \left(\modelout(z; S') - w_z^\top \bm{1}_{S'}\right)^2\,\bigg|\,|\featsupp(S')| = \alpha p
    \right]
    } \\
&\leq \frac{1}{2} \dmerror.
\end{align*}
The case for \cref{eq:expec_lowerbound_2} is analogous.
\end{proof}

Next, we show that $w_z^\top h(\bm{1}_{\featsupp(S)})$ corresponds to the desired difference
of conditional expectations. In this proof, we let $h_\feat = h(\bm{1}_{\featsupp(S)})$ for brevity.

\begin{lemma}
\label{lemma:datamodel-expectation-calculation}
We have for every $x\in S$ that
\begin{align*}
\E_{z\sim \featsupp(S)}\left[
\E_{S' \sim \mathcal{D}_S}
\left[w_z^\top \bm{1}_{S'}\,\bigg|\,|\featsupp(S')| = \alpha\cdot |\featsupp(S)| + 1\right]
-
\E_{S' \sim \mathcal{D}_S}
\left[w_z^\top\bm{1}_{S'}\,\bigg|\,|\featsupp(S')| = \alpha\cdot |\featsupp(S)|\right]
\right]
= \E_{z\sim\featsupp(S)} w_z^\top h_\feat.
\end{align*}
\end{lemma}
\begin{proof}
\label{proof:dm_expec}
Again, let us consider the case for a single example $z$, and let $n=|S|, p=|\featsupp(S)|$.
Then we can write
\begin{align*}
\mathbb{E}_{S' \sim \mathcal{D}_S}\left[w_z^\top \bm{1}_{S'}\,\bigg|\,|\featsupp(S')| = \alpha p\right]
&= \E_{S'\sim \mathcal{D}_S}\left[\sum_{z\in S}\bm{1}_{z\in S'}w_{xz} \bigg|\, |\featsupp(S')| = \alpha p\right] \\
&= \sum_{z\in S}^{}\Pr_{S'\sim \mathcal{D}_S}{\left[z\in S' \bigg|\, |\featsupp(S')| = \alpha p \right]}\cdot w_{xz}.
\end{align*}
There are a total of
\begin{align*}
    {p\choose \alpha p}{n-p\choose \alpha (n-p)}
\end{align*}
sets satisfying $|\featsupp(S')| = \alpha p$. Among these, given that the sample
$z$ contains the feature $\feat$, i.e., $\feat(z) = 1$, there are
\begin{align*}
    {p-1\choose \alpha p-1}{n-p\choose \alpha(n-p)}
\end{align*}
random subsets containing $z$. So for all $z$ containing $\feat$ we have
\begin{align*}
    \Pr_{}{\left[z\in S' \bigg|\, \feat(z)=1, |\featsupp(S')| = \alpha p \right]} = \frac{\alpha p}{p}.
\end{align*}
Similarly, for all samples $z$ that do not contain $\feat$ ,i.e,
$\feat(z) = 0$, we have that
\begin{align*}
    \Pr_{}{\left[z\in S' \bigg|\, \feat(z)=0, |\featsupp(S')| = \alpha p \right]} = \frac{\alpha (n-p)}{n-p}.
\end{align*}
Thus, overall
\begin{align*}
\E_{S' \sim \mathcal{D}_S}\left[w_z^\top \bm{1}_{S'}\,\bigg|\,|\featsupp(S')| = \alpha p\right] & =
\frac{\alpha p}{p}w_z^\top \bm{1}_{\feat(S)} + \frac{\alpha(n-p)}{n-p}w_z^\top (1 - \bm{1}_{\feat(S)}) \nonumber \\
& = w_z^\top \left( \frac{\alpha p}{p} \bm{1}_{\feat(S)} + \frac{\alpha(n-p)}{n-p} (1 - \bm{1}_{\feat(S)}) \right) \\
& = w_z^\top \left( \frac{\alpha p}{p} \bm{1}_{\feat(S)} + \frac{\alpha(n-p)}{n-p} (1-\bm{1}_{\feat(S)}) \right)
\end{align*}
Analogously,
\begin{align*}
\mathbb{E}_{S' \sim \mathcal{D}_S}\left[w_z^\top \bm{1}_{S'}\,\bigg|\,|\featsupp(S')| = \alpha p + 1\right] &=
\frac{\alpha p + 1}{p}w_z^\top \bm{1}_{\feat(S)} + \frac{\alpha(n-p) - 1}{n-p}w_z^\top (1 - \bm{1}_{\feat(S)}) \\
& = w_z^\top \left( \frac{\alpha p+1}{p} \bm{1}_{\feat(S)} + \frac{\alpha(n-p)-1}{n-p} (1 - \bm{1}_{\feat(S)}) \right) \\
& = w_z^\top \left( \frac{\alpha p+1}{p} \bm{1}_{\feat(S)} + \frac{\alpha(n-p)-1}{n-p} (1-\bm{1}_{\feat(S)}) \right)
\end{align*}

We then subtract the two terms to get:
\begin{align*}
    \mathbb{E}_{S' \sim \mathcal{D}_S}\left[w_z^\top \bm{1}_{S'}\,\bigg|\,|\featsupp(S')| = \alpha p + 1\right] - \mathbb{E}_{S' \sim \mathcal{D}_S}\left[w_z^\top \bm{1}_{S'}\,\bigg|\,|\featsupp(S')| = \alpha p \right] \nonumber
    & = w_z^\top \left( \frac{1}{p} \bm{1}_{\feat(S)} - \frac{1}{n-p} (1-\bm{1}_{\feat(S)}) \right) \\
    & = w_z^\top h(\bm{1}_{\phi(S)})
    \end{align*}
Finally, we get the desired results over all examples $z\in\featsupp(S)$ by
directly averaging.
\end{proof}

\begin{proof}[Proof of \cref{lemma:sensitivity-approx}]
The proof of \cref{lemma:sensitivity-approx} follows by combining the results of Lemma
\ref{lemma:datamodel-approximation} and Lemma
\ref{lemma:datamodel-expectation-calculation}:

\begin{align*}
    & \left|
    s_\feat(\alpha \cdot |\featsupp(S)|)
    - \E_{z\sim \featsupp(S)}w_z^\top h(\bm{1}_{\phi(S)})
    \right| \\
    & = \left|
    s_\feat(\alpha \cdot |\featsupp(S)|)
    -
    \E_{z\sim \featsupp(S)} \left[
    \E_{S' \sim \mathcal{D}_S}\left[w_z^\top \bm{1}_{S'}\,\bigg|\,|\featsupp(S')| = \alpha \cdot |\featsupp(S)| + 1\right]
    +\E_{S' \sim \mathcal{D}_S}\left[w_z^\top \bm{1}_{S'}\,\bigg|\,|\featsupp(S')| = \alpha \cdot |\featsupp(S)|\right]
    \right]
    \right| \\
    & \leq \left|
    \E_{z\sim \featsupp(S)} \left[
    \mathbb{E}_{S' \sim \mathcal{D}_S}\left[\modelout(z; S')\,\bigg|\,|\featsupp(S')| = \alpha \cdot |\featsupp(S)|+1\right]
    - \mathbb{E}_{S' \sim \mathcal{D}_S}\left[w_z^\top \bm{1}_{S'}\,\bigg|\,|\featsupp(S')| = \alpha \cdot |\featsupp(S)|+1\right]\right] \right| \\
    & + \left|
    \E_{z\sim \featsupp(S)} \left[
    \mathbb{E}_{S' \sim \mathcal{D}_S}\left[\modelout(z; S')\,\bigg|\,|\featsupp(S')| = \alpha\cdot |\featsupp(S)|\right]
- \mathbb{E}_{S' \sim \mathcal{D}_S}\left[w_z^\top \bm{1}_{S'}\,\bigg|\,|\featsupp(S')| = \alpha \cdot|\featsupp(S)|\right] \right] \right| \\
& \leq 2\cdot\frac{1}{2}\cdot \dmerror = \dmerror
\end{align*}
\end{proof}
    \newpage

    \section{Proof of \cref{lemma:maximizer-is-indicator}}
    \label{app:proof_lem_2}
    \lemmamaxim*

\begin{proof}
    The result follows directly from \cref{ass:strongest_feature} and
    \cref{lemma:sensitivity-approx}. In particular,
    let $\feat_v$ be a feature whose corresponding support $\featsupp_v(S)$ is of size $p$.

\begin{align*}
    h(v)^\top \W v & = \left( \frac{1}{n}\cdot v - \frac{1}{n - p}\cdot (\ones - v) \right)^\top \W v \\
    & =
    \begin{bmatrix}
         h_v^\top w_{z_1} &
         h_v^\top w_{z_2} &
        \ldots &
         h_v^\top w_{z_n}
    \end{bmatrix} \cdot v \\
    & = \sum_{z \in \featsupp_v(S)} h_v^\top w_{z}.
\end{align*}

First, from \cref{lemma:sensitivity-approx}, we have that:
\begin{equation*}
    \sum_{z\in \featsupp_v(S)} h_v^\top w_{z} \leq
     \sum_{z\in \featsupp_v(S)}  s_v(z, \alpha \|v\|_1) + pC^* \varepsilon^{1/2} n^{1/4}
\end{equation*}

Now let $v_p$ be the indicator vector for the poisoned examples. We similarly have from~\cref{lemma:sensitivity-approx}:
\begin{equation*}
    \sum_{z\in \featsupp_p(S)} h_p^\top w_{z} \geq
    \sum_{z\in \featsupp_p(S)}  s_p(z, \alpha \|v\|_1) - pC^* \varepsilon^{1/2} n^{1/4}
\end{equation*}

Thus for any $v\neq v_p$ we have that
\begin{align*}
h_p^\top \W v_p - h_v^\top \W v & = \sum_{z\in \featsupp_p(S)} h_p^\top w_{z} - \sum_{z\in \featsupp_v(S)} h_v^\top w_{z} \\
& \geq
\sum_{z\in \featsupp_{_p}(S)}  s_{v_p}(z, \alpha \|v_p\|_1) -
\sum_{z\in \featsupp_v(S)}  s_v(z, \alpha \|v\|_1)
- 2pC^* \varepsilon^{1/2} n^{1/4}
\end{align*}

We now use~\cref{ass:strongest_feature} that states:

\begin{equation*}
    \sum_{z \in \featsupp_p(S)}
    s_{\feat_p}\left(z, \alpha\cdot p\right)
    -  \sum_{z \in \featsupp(S)} s_{\feat}(z, \alpha\cdot p)
    \geq \delta^*
\end{equation*}

By combining these two inequalities, we obtain:

\begin{align*}
    h_p^\top \W v_p - h_v^\top \W v & = \sum_{z\in \featsupp_p(S)} h_p^\top w_{z} - \sum_{z\in \featsupp_v(S)} h_v^\top w_{z} \\
    & \geq
    \sum_{z\in \featsupp_{_p}(S)}  s_{v_p}(z, \alpha \|v_p\|_1) -
    \sum_{z\in \featsupp_v(S)}  s_v(z, \alpha \|v\|_1)
     - 2pC^* \varepsilon^{1/2} n^{1/4} \\
    & \geq \delta^* - 2pC^* \varepsilon^{1/2} n^{1/4}
\end{align*}

This concludes the proof that the solution of the optimization program in~\cref{eq:maxsum-sub} is the poison indicator vector $v_p$, as long as $\delta^* > 2pC^* \varepsilon^{1/2} n^{1/4}$.

\end{proof}
    \newpage

    \section{Experimental Setup}
    \label{app:exp_setup}

\begin{figure}[htbp]
    \includegraphics[width=\linewidth]{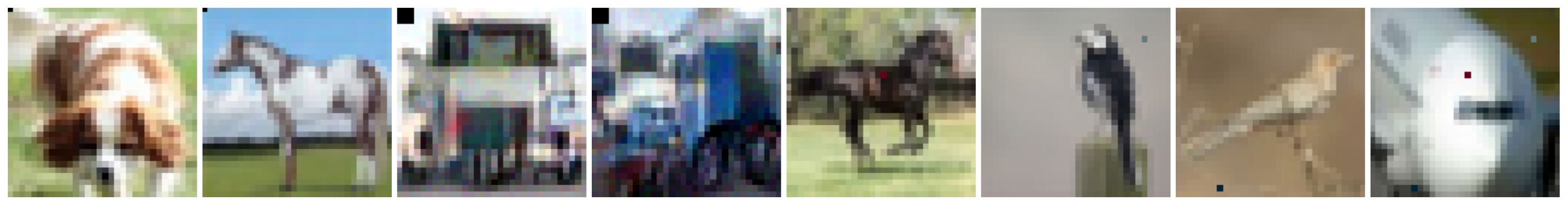}
    \caption{We execute the poisoning attacks with three types of triggers: (a) one black pixel on top left corner (first two images), (b) 3x3 black square on top left corner (third and fourth images), and (c) 3-way triggers adapted from~\citep{xie2020dba} (last four images).}
    \label{fig:triggers}
\end{figure}

\subsection{Backdoor Attacks}

\paragraph{Dirty-Label Backdoor Attacks.} The most prominent type of backdoor attacks is a dirty-label attack~\citep{gu2017badnets}. During a dirty-label attack, the adversary inserts a trigger into a subset of the training set, then flips the label of the poisoned samples to a particular target class $y_b$. We mount four different dirty-label attacks, by considering two different triggers, and two different levels of poisoning (cf. Exp.~1 to 4 in Table~\ref{tab:attack_params}).

\paragraph{Clean-Label Backdoor Attacks.} A more challenging attack is the clean-label attack~\citep{shafahi2018poison, turner2019label}\footnote{We evaluate the clean-label attack as presented in~\citep{turner2019label}} where the adversary avoids changing the label of the poisoned samples. To mount a successful clean-label attack, the adversary poisons samples from the target class only, hoping to create a strong correlation between the target class and the trigger.

We perform two types of clean-label attacks. During the first type (Exp. 5 and 6 from Table~\ref{tab:attack_params}), we perturb the image with an adversarial example before inserting the trigger, as presented in \citep{turner2019label}.
During the second type of clean-label attacks (Exp. 7 and 8 from Table~\ref{tab:attack_params}), we avoid adding the adversarial example, however, we poison more samples to have an effective attack.

\paragraph{Trigger.} We conduct our experiments with two types of triggers. The first type is a fixed pattern inserted in the top left corner of the image. The trigger is unchanged between train and test time. This type of trigger has been used in multiple works~\citep{gu2017badnets,turner2019label}. The other type of trigger is an m-way trigger, with m=3~\citep{xie2020dba}. During training, one of three triggers is chosen at random for each image to be poisoned, and then the trigger is inserted into one of three locations in the image. At test time, all three triggers are inserted at the corresponding positions to reinforce the signal. We display in Figure~\ref{fig:triggers} the triggers used to poison the dataset.

\subsection{Training Setup}

\paragraph{Training CIFAR models.} In this paper, we train a large number of models on different subsets of CIFAR-10 in order to compute the datamodels. To this end, we use the ResNet-9 architecture~\citep{he2015residual}\footnote{\href{https://github.com/wbaek/torchskeleton/blob/master/bin/dawnbench/cifar10.py}{https://github.com/wbaek/torchskeleton/blob/master/bin/dawnbench/cifar10.py}}. This smaller version of ResNets was optimized for fast training.

\paragraph{Training details.}
We fix the training procedure for all our models. We show the hyperparameter details in Table~\ref{tab:hyperparams}\footnote{Our implementation and configuration files will be available in our code.}. One augmentation was used for dirty-label attacks (Cutout \citep{devries2017cutout}) to improve the performance of the model on CIFAR10. Similar to~\citep{turner2019label}, we do not use any data augmentation when performing clean-label attacks.

\begin{table}
    \centering
    \begin{tabular}{@{}ccccccc@{}}
        \toprule
        {\bf Optimizer} & {\bf Epochs} & {\bf Batch Size} & {\bf Peak LR} & {\bf Cyclic LR Peak Epoch} & {\bf Momentum} & {\bf Weight Decay} \\ \midrule
        SGD & 24 & 1,024 & 0 & 5 & 0.9 & 4e-5 \\ \bottomrule
    \end{tabular}
    \caption{Hyperparameters used to train ResNet-9 on CIFAR10.}
    \label{tab:hyperparams}
\end{table}

\paragraph{Performance.} In order to train a large number of models, we use the FFCV library for efficient data-loading~\citep{leclerc2022ffcv}.  The speedup from using FFCV allows us to train a model to convergence in $\sim$40 seconds, and 100k models for each experiment using 16 V100 in roughly 1 day\footnote{We train 3 models in parallel on every V100.}.

\paragraph{Computing datamodels.} We adopt the framework presented in~\citep{ilyas2022datamodels} to compute the datamodels of each experiment. Specifically, we train 100k models on different subsets containing 50\% of the training set chosen at random. We then compute the datamodels as described in~\citep{ilyas2022datamodels}.

\paragraph{Local Search.} We approximate the solution of the problem outlined in \eqref{eq:maxsum-sub} using a local search heuristic presented in \citep{kernighan70}. We iterate over ten sizes for the poison mask: \{10, 20, 40, 80, 160, 320, 640, 1280, 2560, 5120\}. For each size, we collect 1,000 different solutions by starting from different initialization of the solution.

\subsection{Estimating Theoretical Quantities}

Recall the average margin definition presented in \eqref{eq:k-output}. In particular:
\begin{equation}
    g_\feat(k) =
    \mathbb{E}_{z \sim \Phi(S)}\left[
        \mathbb{E}_{S' \sim \mathcal{D}_S}
            \left[f(z; S')\bigg||\featsupp(S')|=k, z \not\in S'\right]
    \right]
\end{equation}
where $S'$ is a subset of the training set, $f(z; S')$ is the margin of the model on example $z$ when trained on the dataset $S'$, $\featsupp(S')$ is the subset of the set $S'$ containing the poisoned feature, and $k$ is the number of poisoned samples. Estimating the average margins requires training a large number of models on different subsets, and measure--for every sample $z$ and every number of poisoned samples $k$--the margins of the trained model.

For the purpose of this paper, we leverage the datamodels computation framework to estimate these averages. In particular, to compute the datamodels weights, we train a large number of models on different subsets $S_1, S_2, \ldots, S_T$ of the training set $S$\footnote{We refer the reader to~\citep{ilyas2022datamodels} for more details.}. For every subset $S_i$, we record the number of poisoned samples in the subset, then we estimate the average margin by averaging the margins over the different subsets that contain $k$ poisoned samples.
\begin{subequations}
    \begin{equation}
        N_\phi(z,k) = \sum_{i=1}^T \bm{1}_{\left(|\featsupp(S_i)|=k \right) \, \land \, \left( z \notin S_i\right)}
    \end{equation}

    \begin{equation}
        \mathbb{E}_{S' \sim \mathcal{D}_S}
            \left[f(z; S')\bigg||\featsupp(S')|=k, z \not\in S'\right] \approx \frac{1}{N_\phi(k)} \sum_{i=1}^T f(z; S_i) \cdot \bm{1}_{\left(|\featsupp(S_i)|=k \right) \, \land \, \left( z \notin S_i\right)}
    \end{equation}

    \begin{equation}
        g_\feat(k) \approx \frac{1}{|\featsupp(S)|} \sum_{z\in \featsupp(S)} \frac{1}{N_\phi(z,k)} \sum_{i=1}^T f(z; S_i) \cdot \bm{1}_{\left(|\featsupp(S_i)|=k \right) \, \land \, \left( z \notin S_i\right)}
    \end{equation}

\end{subequations}
By a training 100k models on different subsets of the dataset, we obtain reasonable estimates of the marginal effects for every sample $z=(x,y)$.
    \newpage

    \section{Omitted Plots}
    \label{app:plots}
    \subsection{Average Margin Plots}
\label{app:more_margin}
In this section, we show for all the experiments the plots of the average margin for clean and poisoned samples as a function of the number of poisoned samples in the dataset (cf. Fig.~\ref{fig:margin-lines} in the main paper).

\begin{figure}[h]

    \includegraphics[width=\linewidth]{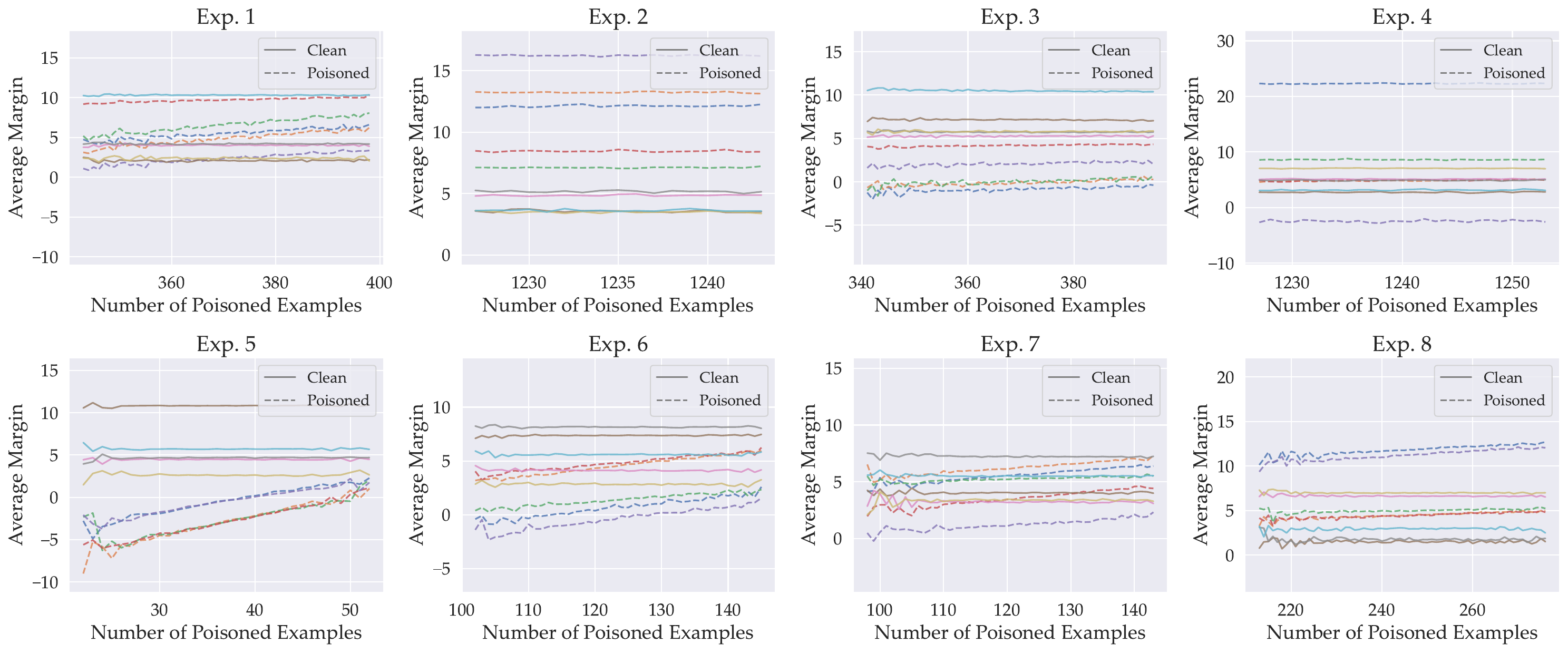}

    \caption{We plot for all the experiments the average margin for five clean samples (left) and five poisoned samples (right) as the number of poisoned samples in the training set increases.
    We observe that the average margin for {\it clean samples} (without the trigger) is constant when poisoning
    more samples in the dataset. In contrast, the average margin for {\it poisoned samples} (with the trigger)
    increases when the number of poisoned samples increases in the dataset, confirming our assumptions.}
    \label{fig:all-margin-lines}
\end{figure}

\subsection{Estimated Sensitivities Plots}

In this section, we show for all the experiments the plots of the estimated sensitivities, and the approximation we obtain using datamodels (cf. Fig.~\ref{fig:theory-sensitivity} in the main paper).

\begin{figure}[h]

    \includegraphics[width=\linewidth]{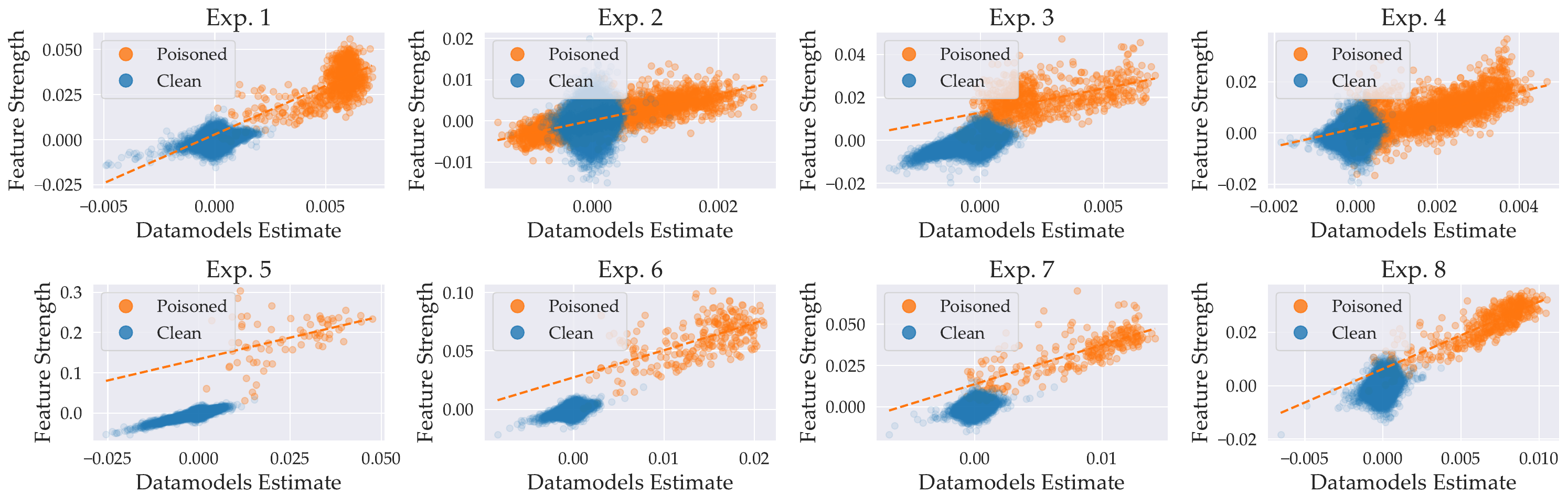}

    \caption{We plot for all the experiments the estimated marginal sensitivities and the approximation with datamodels presented in Equation~\ref{eq:approx-sensitivity}. We observe for poisoned samples (in red) a good linear correlation between the sensitivities and the datamodels' approximation. Additionally, we observe no noticeable correlation for clean samples (in green).}
    \label{fig:all-theory-sensistivity}
\end{figure}

\subsection{Distribution of Datamodels Values}

In this section, we plot for each experiment the distribution of the datamodels weights for all experiments. In particular, recall that the datamodels weight $w_x[i]$ represents the influence of the training sample $x_i$ on the sample $x$. We show below the distribution of the effect of 1) poisoned samples on poisoned samples, 2) the poisoned samples on the clean samples and 4) the clean samples on the clean samples.

\begin{figure}[h]

    \includegraphics[width=\linewidth]{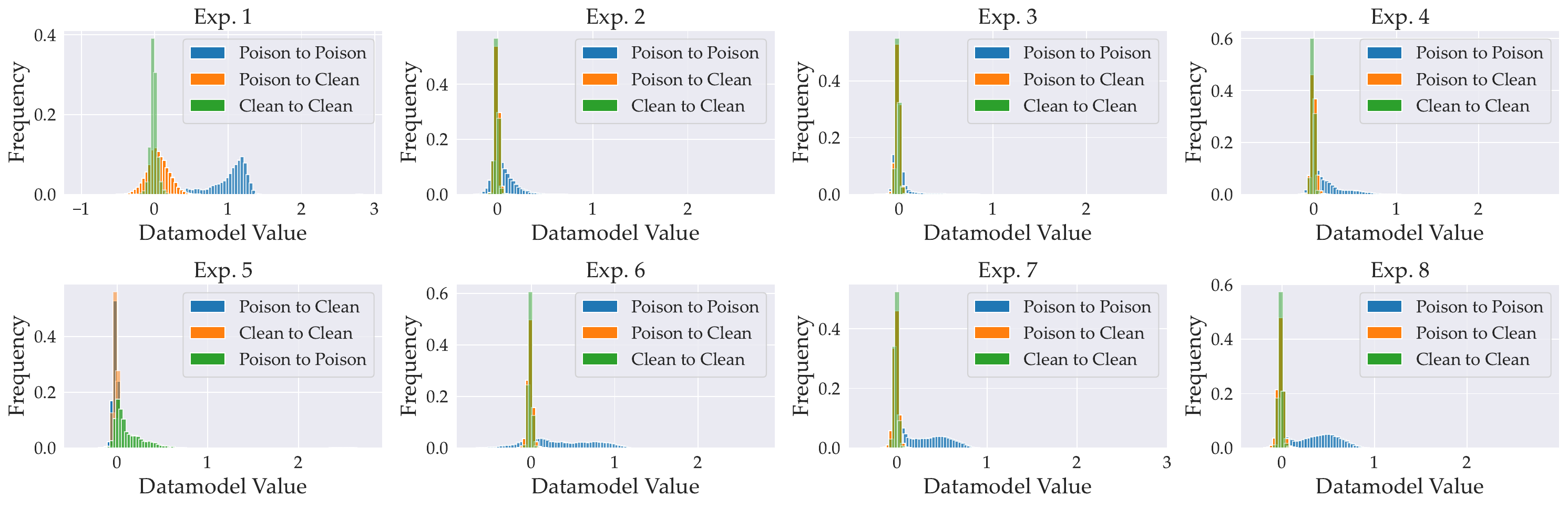}

    \caption{We plot the distribution of the datamodels weights for all the experiments. We clearly see that the effect of poisoned samples on other poisoned samples is generally higher than the effect of poisoned samples on clean samples, and than clean samples on each other.}
    \label{fig:all-dm-hist}
\end{figure}

\subsection{Attack Success Rate (ASR)}

In the main paper, we presented our results by measuring the accuracy of a model on a clean and a poisoned validation sets. Another relevant metric is the Attack Success Rate (ASR) which measures the probability that the model predicts the target class after inserting the trigger into an image. As we can see in Table~\ref{tab:results_asr}, our defense leads to a low ASR in seven out of eight setups.

\begin{table}[h]

        \centering
        \caption{We compare our method against multiple baselines in a wide range of experiments. We observe that our algorithm leads to a low ASR in seven out of eight experiments. Refer to Table~\ref{tab:attack_params} for the full experiments parameters.}
            \begin{tabular}{@{}ccccccc@{}}
                \toprule
                \textbf{Exp.}          & \textbf{No Defense}        & \textbf{AC}                & \textbf{ISPL}              & \textbf{SPECTRE}           & \textbf{SS}                & \textbf{Ours} \\ \midrule
                1 & 87.94 & 88.26 & {\bf0.70}  & 87.67 & 73.78 & {\bf0.81}         \\
                2 & 96.38 & 96.32 & {\bf0.67}  & -     & 95.40 & {\bf 1.44}          \\ \midrule
                3 & 50.49 & 51.33 & {\bf0.58}  & 55.68 & 10.44 & {\bf1.18}       \\
                4 & 99.21 & 99.02 & {\bf0.75}  & -     & 95.85 & {\bf2.30}          \\ \midrule
                5 & 15.66 & 5.35  & {\bf0.71}  & 7.66  & {\bf0.80}  & {\bf0.92}       \\
                6 & 58.57 & 45.44 & {\bf0.66}  & 46.78 & {\bf0.67}  & {\bf0.77}          \\ \midrule
                7 & 26.09 & 23.64 & 9.90   & 18.48 & 9.17  & {\bf3.56}       \\
                8 & 50.62 & 44.82 & 26.07 & 44.14 & 24.72 & {\bf3.42}         \\ \bottomrule
                \end{tabular}

    \label{tab:results_asr}
\end{table}

    \newpage

\end{document}